\documentclass[11pt,oneside]{article}

\usepackage{amsthm, amsmath, amssymb, amsfonts, graphicx, epsfig, mathtools, epstopdf}
\usepackage[section]{placeins}
\usepackage{morefloats}
\usepackage{float}
\usepackage{bbm}

\usepackage{enumerate}

\usepackage[normalem]{ulem}

\usepackage[colorlinks=true, pdfstartview=FitV, linkcolor=blue,
            citecolor=blue, urlcolor=blue]{hyperref}
\usepackage[usenames]{color}
\definecolor{Red}{rgb}{0.7,0,0.1}
\definecolor{Green}{rgb}{0,0.7,0}

\usepackage{url}
\makeatletter\def\url@leostyle{%
 \@ifundefined{selectfont}{\def\UrlFont{\sf}}{\def\UrlFont{\scriptsize\ttfamily}}} \makeatother

\newtheorem{theorem}{Theorem}[section]
\newtheorem{proposition}[theorem]{Proposition}
\newtheorem{lemma}{Lemma}[section]
\newtheorem{corollary}[theorem]{Corollary}

\theoremstyle{definition}
\newtheorem{definition}{Definition}[section]

\theoremstyle{remark}
\newtheorem{remark}[theorem]{Remark}

\newtheoremstyle{dotless}{}{}{\itshape}{}{\bfseries}{}{ }{}
\theoremstyle{dotless}
\newtheorem{assumption}{Assumption}

\def\cD{\mathcal{D}}

\def\cF{\mathcal{F}}
\def\cG{\mathcal{G}}

\def\cK{\mathcal{K}}

\def\cS{\mathcal{S}}

\def\cX{\mathcal{X}}

\def\bE{\mathbb{E}}
\def\bF{\mathbb{F}}
\def\bG{\mathbb{G}}

\def\bP{\mathbb{P}}
\def\bQ{\mathbb{Q}}
\def\bR{\mathbb{R}}

\def\bX{\mathbb{X}}

\def\1{\mathbbm{1}}

\title{Collateralized CVA Valuation with Rating Triggers and Credit Migrations}
\author{Tomasz R. Bielecki\footnote{Tomasz R. Bielecki and Igor Cialenco acknowledge support from the NSF grant DMS-0908099.}\\
\small{Department of Applied Mathematics,}\\[-0.3ex]
\small{Illinois Institute of Technology,}\\[-0.3ex]
\small{Chicago, 60616 IL, USA}\\[-0.3ex]
\url{bielecki@iit.edu}\\
\and
Igor Cialenco\footnotemark[\value{footnote}]\\[-0.3ex]
\small{Department of Applied Mathematics,}\\[-0.3ex]
\small{Illinois Institute of Technology,}\\[-0.3ex]
\small{Chicago, 60616 IL, USA}\\[-0.3ex]
\url{igor@math.iit.edu} \\
\and
Ismail Iyigunler \\
\small{Department of Applied Mathematics,}\\[-0.3ex]
\small{Illinois Institute of Technology,}\\[-0.3ex]
\small{Chicago, 60616 IL, USA}\\[-0.3ex]
\url{iiyigunl@iit.edu} \\}

\date{First Circulated: May 12, 2012}

\begin{document}
\maketitle

\begin{abstract}

\noindent In this paper we discuss the issue of computation of the bilateral credit valuation adjustment (CVA) under rating triggers, and  in presence of ratings-linked margin agreements. Specifically, we consider collateralized OTC contracts, that are subject to rating triggers, between two parties -- an investor and a counterparty. Moreover, we model the margin process as a functional of the credit ratings of the counterparty and the investor.
We employ a Markovian approach for modeling of the rating transitions of the  two parties to the contract.
In this framework, we derive the representation for bilateral CVA. We also introduce a new component  in the decomposition of the counterparty risky price: namely the rating valuation adjustment (RVA) that accounts for the rating triggers. We give two examples of dynamic collateralization schemes where the margin thresholds are linked to the credit ratings  of the parties. We account for the rehypothecation risk in the presence of independent amounts. Our results are illustrated via computation of various counterparty risk adjustments for a CDS contract and for an IRS contract.
\end{abstract}

{\noindent \small
}

\newpage
\section{Introduction}
Modeling, managing and mitigating counterparty risk is a crucial task for all financial institutions.
One of the most popular mitigation techniques used by the market participants is including additional termination events (ATE) in OTC transactions. As defined in Section 5(b)(vi) of the ISDA Master Agreement (see \cite{ISDA2002}), ATEs allow institutions to terminate and close out the derivatives transactions with a counterparty if a termination event occurs.
We consider a particular, and in fact the most common, termination event: rating triggers.

A rating trigger is defined as a threshold credit rating level, which is agreed upon the initiation of the contract.
If the credit rating of the counterparty or the investor decreases below the trigger level, before the maturity of the contract, the contract is terminated and closed out.
Therefore, rating triggers provide additional protection from a counterparty with a deteriorating credit rating, by allowing the investor to terminate the contract prior to a default event.
Furthermore, since a significant credit deterioration is usually followed by a default event, adding rating triggers serves as a cushion against any losses resulting from such defaults.
On the other hand, rating-triggers are also very effective in mitigating counterparty credit risk.

Counterparty risk modeling has gained paramount importance since the credit crisis in 2008.
As it is noted in Benjamin \cite{Benjamin2010}, just shy of one-thirds of the losses in the crisis were actually due to realized default events, whereas about two-thirds were due to mark-to-market losses associated with the counterparty credit risk.
Naturally, counterparty credit risk modeling literature has grown significantly.
We refer to Bielecki, Cialenco, and Iyigunler \cite{BCI2011}, Assefa, Bielecki, Crepey and Jeanblanc \cite{Assefa2011}, Brigo, Capponi, Pallavicini and Papatheodorou \cite{BCPP2011} and also Crepey \cite{Crepey12funding1,Crepey12funding2} for recent general results in counterparty risk modeling.

On the contrary, the literature on counterparty risk modeling with rating triggers is very limited.
In Yi \cite{Yi2011}, CVA valuation with rating triggers is studied for optional and mandatory termination events, and a compound Poisson model is introduced for modeling rating transitions and default probabilities.
Zhou \cite{Zhou2010} considers practical problems regarding CVA valuation with additional termination events is considered under a simple model from a practitioner's point of view.
Recently, Mercurio \cite{Mercurio2011} studied a similar problem and introduced a valuation model by proposing several generalizing assumptions to simplify the CVA computations considering the unilateral counterparty risk.
However, a comprehensive approach which involves the joint modeling of rating transitions in a risk-neutral setting and the dynamic, ratings-dependent collateralization has not been studied in the literature.

In this paper we consider the problem of collateralized bilateral CVA valuation with rating triggers and credit migrations.
We first find the CVA representation in presence of rating triggers. We show that the value of the underlying OTC contract needs to be adjusted also for the rating triggers.
This new adjustment term is called as the rating valuation adjustment (RVA). We show that RVA represents the expected loss in case of a default event that is preceded by a trigger event. In the bilateral case, we see that RVA is decomposed into two components: URVA and DRVA, representing the rating valuation adjustments for the counterparty's and the investor's rating triggers.
Furthermore, we consider dynamic collateralization using the rating transitions.
In this framework, the collateral thresholds are defined as the functionals of the current credit ratings of the counterparty and the investor.
In practice such rating-dependent margin agreements are standard and they are described in the Credit Support Annex (CSA).
Moreover, we consider the rehypothecation risk of the collateral in the presence of independent amounts.

We employ the Markovian copula approach for modeling the joint rating transitions and the default probabilities of the counterparty and the investor.
Previously, Bielecki, Vidozzi, and Vidozzi \cite{BieleckiEfficient2006,Bielecki2008a} applied Markov copulae to the collateralized debt obligations and ratings-triggered corporate step-up bonds.
Theoretical aspects of the Markov copulae can be found in Bielecki, Jakubowski, Vidozzi, and Vidozzi \cite{BieleckiJVV2008} and  Bielecki, Jakubowski, and Nieweglowski \cite{BJN2011}.
We  illustrate our results with numerical examples. We analyze the impact of early termination clauses and dynamic collateralization on the bilateral and unilateral CVA, as well as the DVA and RVA in case of a CDS and an IRS contract.

This paper is organized as follows. In Section \ref{section:CVAratings} we present a general framework for the valuation of collateralized credit valuation adjustment in the presence of rating triggers. We study dynamic collateralization in Section \ref{section:collateral}, and we study rehypothecation in Section \ref{section:rehypothecation}. Next, we employ the Markovian copula approach for modeling the joint rating transitions of the counterparty and the investor in Section \ref{section:Copulae}. Finally, we present numerical results in case of a CDS contract and an IRS contract in Section \ref{section:Apps}.

\section{Credit Value Adjustment and Collateralization under Rating Triggers}\label{section:CVAratings}

We consider an OTC contract between two names: the investor and the counterparty. In the model we propose in this paper, the counterparty risk associated with this contract will be sensitive to the current credit worthiness of the two parties. We postulate that the creditworthiness of each party is represented by the same $\cK:=\{ 1,2,\ldots ,K\}$ rating categories. We postulate that the ratings are ordered from the best, i.e. $1$, to the worst, i.e. $K$, with the convention that the level $K$ corresponds to a default.

 To model the evolution of the credit worthiness we introduce two right continuous processes $X^1$ and $X^2$ on $(\Omega ,\cG, \bQ)$, with values in $\cK$, and we denote by $\bX^1$ and $\bX^2$ the associated filtrations: $\bX^i=\left (\cX^i_t\right)_{t \geq 0}$ with $\cX^i_t= \sigma(X^i_u,u \leq t)$ for $t \in \bR_+$, $i=1,2$.\footnote{All filtrations considered in this paper are assumed to satisfy the usual conditions.}
 Processes $X^1$ and $X^2$ represent the evolution of the credit ratings of the counterparty and the investor.
 In what follows we shall make additional specific assumptions about processes $X^i,\ i=1,2.$

We assume that we are given a market filtration $\bF$ containing the information about the relevant market variables (i.e. short rate process), and a filtration $\widetilde\bF$ that contains the information regarding the financial contract underlying  our OTC contract. Accordingly, we define $\bG := \bF \vee \widetilde\bF \vee \bX^1 \vee \bX^2$.

The savings account process $B$ is given as,
\begin{equation*}
B_t :=e^{\int_0^t r_s ds} \, , \quad t \in [0,T] \, ,
\end{equation*}
where the $\bF$-adapted process $r$ models the short-term interest rate. We postulate that $\bQ$ represents a pricing measure corresponding to the discount factor $\beta = B^{-1}.$

As noticed above, the two parties are default prone, and the respective default times are given as
\begin{align*} \label{eq:tau}
\tau_i :=\inf \{ t>0:X_t^i=K\} \, , \quad i=1,2.
\end{align*}
We shall also consider the first default time  $\tau:=\tau_1\wedge \tau_2$. 

We denote by $R_i\in [0,1]$ a $\cG_{\tau_i}$--measurable random variable, which represents the recovery rate of party $i=1,2$.
In our model, the recovery rates represent the fraction of the mark-to-market value of the underlying contract recovered from the defaulting names, which appears in the close-out amounts.

We denote by $\bE_{\rho}$ the conditional expectation given $\cG_{\rho}$ under $\bQ$ for any $\bG$ stopping time $\rho$.

Let $D$ represent the counterparty risk-free cumulative dividend process of our OTC contract over a finite time horizon $[0,T]$, which is the ``clean'' version of the contract that does not account for the counterparty risk.\footnote{All cash flows are considered from the point of view of the investor.} We assume that $D$ is of finite variation.

In accordance with the classical risk neutral valuation we define the counterparty risk-free ex-dividend price (\emph{clean price} from now on) process of the contract:

\begin{definition}
The ex-dividend price process of a counterparty risk-free contract is defined as,
\footnote{Required integrability properties are assumed implicitly.}
\begin{equation*}
S_t=B_t\bE_t\Big[\int_{]{t},T]}B_u^{-1}dD_u\Big] \, ,
\end{equation*}
for all $t \in [0,T]$.
\end{definition}

\noindent
The process $S$ is also called the clean mark-to-market process. Let us also define the process $S^{\Delta}:=S + \Delta D$.

We consider collateralized contracts, therefore we define a $\bG$-predictable process $C$ on $[0,T]$ representing cumulative collateral amount in the margin account.
Mechanics and the modeling of the collateral process are discussed in Section \ref{section:collateral}.

\subsection{Pricing Bilateral Counterparty Risk}\label{sec:CVAnorating}
Let us consider the case $K =2$, therefore allowing only the default and the pre-default states prevail. This case corresponds to the models presented and discussed by Bielecki et al. \cite{BCI2011,Bielecki2011a,Bielecki2011} and Assefa et al. \cite{Assefa2011}.

We denote $H^1_t:=\1_{\{\tau_1 \leq t\}}$ and $H^2_t:=\1_{\{\tau_2 \leq t\}}$ as the default indicator processes of $\tau_1$ and $\tau_2$ respectively. Note that we now have $X_t^1=1+H^1_t$ and $X_t^2=1+H^2_t$.
We also define $\tau := \tau_1 \wedge \tau_2$ as the first default time of the counterparty and the investor. Moreover, we let $H:=\1_{\{\tau \leq t\}}$ be the default indicator process corresponding to $\tau$.

Let $\cD^C$ represent the counterparty risky cumulative dividend processes of the contract that is subject to counterparty default risk. Therefore, given $D$, we define the counterparty risky cumulative dividend process $\cD^C$ as follows.

\begin{definition}\label{def:cumDivNORating}
Counterparty-risky cumulative dividend process has the following form
\begin{align*}
\cD_t &=(1-H_t)D_t+H_tD_{\tau -}+\1_{\{\tau <T\}}(C_{\tau }\1_{\{\tau \leq t\}}+(R_1(S^{\Delta}_{\tau }-C_{\tau})^{+}
-(S^{\Delta}_{\tau }-C_{\tau })^{-})\1_{\{\tau =\tau_1 \leq t\}}  \\[0.05in]
\ \ \ &-(R_2(S^{\Delta}_{\tau }-C_{\tau })^{-}-(S^{\Delta}_{\tau}-C_{\tau })^{+})\1_{\{\tau =\tau_2 \leq t\}}
- (S^{\Delta}_{\tau }-C_{\tau })\1_{\{\tau = \tau_1 = \tau_2 \leq t\}}) \, ,
\end{align*}
for all $t\in[0,T]$.
\end{definition}

We proceed with defining the ex-dividend price of a counterparty risky contract,

\begin{definition}
The ex-dividend price process of a counterparty risky contract is given as,
\begin{equation*}
\cS_t=B_t\bE_t\Big[\int_{]{t},T]}B_u^{-1}d\cD_u \Big] \, ,
\end{equation*}
for all $t \in [0,T]$.
\end{definition}

Having defined a counterparty risk-free and a counterparty risky contracts, we are now interested in the difference between their ex-dividend prices. This difference is called as the \emph{Credit Valuation Adjustment} (CVA). Since we consider bilateral case, both the investor and the counterparty can default on their contractual obligations. Therefore, we refer to the CVA as the \emph{bilateral} credit valuation adjustment.
\begin{definition}\label{CVAnoRating}
The bilateral credit valuation adjustment is defined as,
\begin{equation*}
\text{CVA}_t=S_t-\cS_t \, ,
\end{equation*}
for all $t \in [0,\tau \wedge T]$.
\end{definition}

Bilateral counterparty valuation adjustment process has the following representation.
\begin{proposition}
The credit valuation adjustment process can be represented as
\begin{align}\label{CVA}
\text{CVA}_t = &B_t\bE_t\Big[\1_{\{ \tau =\tau_1\leq T\}}B_{\tau}^{-1}(1-R_1)(S^{\Delta}_{\tau}-C_{\tau })^{+}\Big] \nonumber \\[0.05in]
-&B_t\bE_t\Big[\1_{\{\tau =\tau_2\leq T\}}B_{\tau}^{-1}(1-R_2)(S^{\Delta}_{\tau}-C_{\tau })^{-}\Big] \, ,
\end{align}
for all $t \in [0,\tau \wedge T]$.
\end{proposition}
A proof of this proposition, where the underlying is assumed to be a CDS contract, can be found in Bielecki et al. \cite{BCI2011}. Note that the bilateral CVA can be decomposed as
\begin{align}
\text{UCVA}_t&=B_t\bE_t[\1_{\{ \tau =\tau_1\leq T\}}B_{\tau }^{-1}(1-R_1)(S^{\Delta}_{\tau}-C_{\tau })^{+}] \, , \label{eq:DVA} \\[0.05in]
\text{DVA}_t&=B_t\bE_t[\1_{\{ \tau =\tau_2\leq T\}}B_{\tau }^{-1}(1-R_2)(S^{\Delta}_{\tau}-C_{\tau })^{-}] \nonumber \, ,
\end{align}
for all $t \in [0,\tau \wedge T]$.
These components represent the two legs of bilateral CVA, namely the Unilateral Credit Valuation Adjustment (UCVA) and the Debt Valuation Adjustment (DVA), representing the expected losses in case of the counterparty's and the investor's defaults, respectively.
\subsection{Pricing Bilateral Counterparty Risk with Rating Triggers}\label{sec:CVArating}
We now proceed with introducing the rating trigger times, and the close-out cash flows in the CVA valuation. We also show how the clean price of our OTC contract can be adjusted for the counterparty risk and the rating triggers.

\subsubsection{Trigger Times}

As we already said, the counterparty risk of the OTC contract that we study in this paper is sensitive to the creditworthiness of the investor and the counterparty.
Specifically, we consider an OTC contract that is subject to a \textit{rating trigger clause}:
\begin{quote}
If the investor's or the counterparty's credit rating deteriorates to or below the \textit{trigger level} (except the default level), the contract is terminated and closed out.
Note that there are no MtM losses associated with the trigger events.
\end{quote}

The trigger levels are set as $K_1$ for the counterparty, and $K_2$ for the investor,\footnote{It is implicitly assumed that $X^i_0<K_i$ for $i=1,2$.}
where $1<K_1,K_2\leq K.$ Let $\tau_i^R$ represent the first time that the i-th party's credit rating crosses the his rating trigger, that is
\begin{align*}
\tau_i^R &:=\inf \{ {t>0:X_t^i\geq K_i\}} \, , \quad i=1,2.
\end{align*}
The corresponding rating trigger event times\footnote{That is, excluding default.} are defined as
\begin{align*}
\widehat{\tau_i}^{R} &:=\inf \{ {t>0:X_t^i \in \{ K_i,K_{i+1},\dots,K-1 \} \}} \, , \quad i=1,2,
\end{align*}
and we set $$\tau^R:=\tau_1^R \wedge \tau_2^R \quad \text{and} \quad \widehat{\tau}^{R}:=\widehat{\tau_1}^{R} \wedge \widehat{\tau_2}^{R} \, .$$
Clearly, $\tau_i^R = \widehat{\tau_i}^{R} \wedge \tau_i$ for $i=1,2$.

We denote by $H^R_t:=\1 _{\{\tau^R \leq t\}}$ and $\widehat{H}^R_t:=\1_{\{\widehat{\tau}^R \leq t\}}$ the rating trigger indicator processes  including and not including the default event, respectively.

\subsubsection{Cash Flows, Prices and Adjustments}\label{subsec:dividRat}

The close-out portion of the cumulative dividend process of the counterparty risky contract needs to account for the MtM exchange without incurring any losses at a trigger time other than  default. On the other hand,
if a trigger event occurs simultaneously with a default event, the deal will be settled according to the default event.
Consequently, we propose the following definition of the cumulative dividend process of the counterparty risky contract,
\begin{definition}\label{def:cumDivRating}
The counterparty-risky cumulative dividend process of an OTC contract subject to rating triggers is defined as
\begin{align*}
{\cD}^R_t &=(1-H^R_t)D_t
+D_{\tau^R-}H^R_t + \1_{\{\tau^R \leq T\}}\Big( C_{\tau^R}H^R_t \\[0.05in]
& \quad +\left(R_1(S^{\Delta}_{\tau^R} -C_{\tau^R})^{+}-(S^{\Delta}_{\tau^R} -C_{\tau^R})^{-}\right)[H^R,H^1]_t  \\[0.05in]
& \quad -\left(R_2(S^{\Delta}_{\tau^R} -C_{\tau^R})^{-}-(S^{\Delta}_{\tau^R} -C_{\tau^R})^{+}\right)[H^R,H^2]_t \\[0.05in]
& \quad -\left(S^{\Delta}_{\tau^R} -C_{\tau^R}\right)[[H^R,H^1] \, ,H^2]_t
+\left(S^{\Delta}_{\tau^R} -C_{\tau^R}\right)[H^R,\widehat{H}^{R}]_t \Big) \, ,
\end{align*}
for all $t\in[0,T]$.
\end{definition}

Accordingly, the ex-dividend price processes associated with a counterparty risky contract with rating triggers is defined as follows,

\begin{definition}\label{def:CDSwRT}
The ex-dividend price process ${S}_t^R$ of a counterparty risky contract with rating triggers, maturing at time $T$, is defined as
\begin{equation*}
{S}_t^R=B_t\bE_t \Big[ \int_{]t,T]} B_u^{-1} d {\cD}^R_u \Big] \, ,
\end{equation*}
for all $t\in[0,T]$.
\end{definition}

We now introduce the credit valuation adjustment term when the underlying contract is subject to rating triggers.
\begin{definition}
The bilateral credit valuation adjustment with rating triggers is defined as,
\begin{equation}\label{CVAR}
\text{CVA}_t^R=S_t-S_t^R \, ,
\end{equation}
for $t\in[0,\tau^R \wedge T]$.
\end{definition}

The following representation generalizes the results derived in Bielecki et al. \cite{BCI2011}.
\begin{proposition}
The bilateral credit valuation adjustment defined in \eqref{CVAR} can be represented as
\begin{align}\label{CVAR-1}
\text{CVA}_t^R = &B_t\bE_t\Big[\1_{\{\tau^R=\tau_1\leq T\}}B_{\tau^R}^{-1}(1-R_1)(S^{\Delta}_{\tau^R }-C_{\tau^R})^{+}\Big] \nonumber \\[0.05in]
-&B_t\bE_t\Big[\1_{\{ \tau^R=\tau_2\leq T\}}B_{\tau^R}^{-1}(1-R_2)(S^{\Delta}_{\tau^R } -C_{\tau^R})^{-}\Big] \, ,
\end{align}
for $t\in[0,\tau^R \wedge T]$.
\end{proposition}

\begin{proof}
Using Definition \ref{def:cumDivRating}, we get,
\begin{align*}
  dD_t-d{\cD}^R_t &= dD_t- (1-H^R_t)dD_t -D_{t-}dH^R_t +D_{\tau^R-}dH^R_t - \1_{\{\tau^R \leq T\}}C_{\tau^R}dH^R_t  \\[0.05in]
& \quad -\1_{\{\tau^R \leq T\}}\big(R_1(S^{\Delta}_{\tau^R}-C_{\tau^R})^{+}-(S^{\Delta}_{\tau^R}-C_{\tau^R})^{-}\big)d[H^R,H^1]_t \\[0.05in]
& \quad +\1_{\{\tau^R \leq T\}}\big(R_2(S^{\Delta}_{\tau^R}-C_{\tau^R})^{-} -(S^{\Delta}_{\tau^R}-C_{\tau^R})^{+}\big)d[H^R,H^2]_t \\[0.05in]
& \quad +\1_{\{\tau^R \leq T\}}(S^{\Delta}_{\tau^R}-C_{\tau^R})d[[H^R,H^1] \, ,H^2]_t  \\[0.05in]
& \quad -\1_{\{\tau^R \leq T\}}(S^{\Delta}_{\tau^R}-C_{\tau^R})d[H^R,\widehat{H}^R]_t \, ,
\end{align*}
Integrating both sides leads to,
\begin{align*}
  \int_{]{t},T]}B_u^{-1}(dD_u-d{\cD}^R_u) & = \int_{]{t},T]}B_u^{-1}H^R_udD_u - \int_{]{t},T]}B_u^{-1}D_{u-}dH^R_u + \int_{]{t},T]}B_u^{-1}D_{\tau^R-}dH^R_u \\[0.05in]
& \hspace{-2cm} -\int_{]{t},T]}\1_{\{\tau^R \leq T\}}B_u^{-1}\big(R_1(S^{\Delta}_{\tau^R}-C_{\tau^R})^{+}  -(S^{\Delta}_{\tau^R}-C_{\tau^R})^{-}\big)d[H^R,H^1]_u \\[0.05in]
& \hspace{-2cm} +\int_{]{t},T]}\1_{\{\tau^R \leq T\}}B_u^{-1}\big(R_2(S^{\Delta}_{\tau^R}-C_{\tau^R})^{-} -(S^{\Delta}_{\tau^R}-C_{\tau^R})^{+}\big)d[H^R,H^2]_u \\[0.05in]
& \hspace{-2cm} +\int_{]{t},T]}\1_{\{\tau^R \leq T\}}B_u^{-1}(S^{\Delta}_{\tau^R}-C_{\tau^R})d[[H^R,H^1] \, ,H^2]_u  \\[0.05in]
& \hspace{-2cm} -\int_{]{t},T]}\1_{\{\tau^R \leq T\}}B_u^{-1}(S^{\Delta}_{\tau^R}-C_{\tau^R})d[H^R,\widehat{H}^R]_u- \int_{]{t},T]}\1_{\{\tau^R \leq T\}}B_u^{-1}C_{\tau^R}dH^R_u \, .
\end{align*}
Since,
\begin{equation*}
   \int_{]{t},T]}B_u^{-1}D_{\tau^R-}dH^R_u - \int_{]{t},T]}B_u^{-1}D_{u-}dH^R_u = 0 \, ,
\end{equation*}
we obtain,
\begin{align*}
  \int_{]{t},T]}B_u^{-1}(dD_u-d{\cD}^R_u) & = \int_{]{t},T]}B_u^{-1}H^R_udD_u - \int_{]{t},T]}\1_{\{\tau^R \leq T\}}B_u^{-2}C_{\tau^R}dH^R_u \\[0.05in]
& \hspace{-2cm} -\int_{]{t},T]}\1_{\{\tau^R \leq T\}}B_u^{-1}\big(R_1(S^{\Delta}_{\tau^R}-C_{\tau^R})^{+} -(S^{\Delta}_{\tau^R}-C_{\tau^R})^{-}\big)d[H^R,H^1]_u \\[0.05in]
& \hspace{-2cm} +\int_{]{t},T]}\1_{\{\tau^R \leq T\}}B_u^{-1}\big(R_2(S^{\Delta}_{\tau^R}-C_{\tau^R})^{-}  -(S^{\Delta}_{\tau^R}-C_{\tau^R})^{+}\big)d[H^R,H^2]_u \\[0.05in]
& \hspace{-2cm} +\int_{]{t},T]}\1_{\{\tau^R \leq T\}}B_u^{-1}(S^{\Delta}_{\tau^R}-C_{\tau^R})d[[H^R,H^1] \, ,H^2]_u  \\[0.05in]
& \hspace{-2cm} -\int_{]{t},T]}\1_{\{\tau^R \leq T\}}B_u^{-1}(S^{\Delta}_{\tau^R}-C_{\tau^R})d[H^R,\widehat{H}^R]_u \, .
\end{align*}
Conditioning on $\tau^R$, we get
\begin{align}
 \1_{\{t \leq  \tau^R \wedge T\}}\ \bE_{\tau^R} \Big[ \int_{]{t},T]}B_u^{-1}(dD_u-d{\cD}^R_u) \Big]
 &  =\1_{\{t \leq  \tau^R \wedge T\}}\ \bE_{\tau^R} \Bigg[ \int_{]{t},T]}B_u^{-1}H^R_udD_u \nonumber \\[0.05in]
  & \hspace{-2cm} - \int_{]{t},T]}\1_{\{\tau^R \leq T\}}B_u^{-1}C_{\tau^R}dH^R_u \nonumber \\[0.05in]
& \hspace{-2cm} -\int_{]{t},T]}\1_{\{\tau^R \leq T\}}B_u^{-1}\big(R_1(S^{\Delta}_{\tau^R}-C_{\tau^R})^{+}  -(S^{\Delta}_{\tau^R}-C_{\tau^R})^{-}\big)d[H^R,H^1]_u \nonumber \\[0.05in]
& \hspace{-2cm} +\int_{]{t},T]}\1_{\{\tau^R \leq T\}}B_u^{-1}\big(R_2(S^{\Delta}_{\tau^R}-C_{\tau^R})^{-} -(S^{\Delta}_{\tau^R}-C_{\tau^R})^{+}\big)d[H^R,H^2]_u \nonumber \\[0.05in]
& \hspace{-2cm} +\int_{]{t},T]}\1_{\{\tau^R \leq T\}}B_u^{-1}(S^{\Delta}_{\tau^R}-C_{\tau^R})d[[H^R,H^1] \, ,H^2]_u \label{eq:condTau} \\[0.05in]
& \hspace{-2cm} -\int_{]{t},T]}\1_{\{\tau^R \leq T\}}B_u^{-1}(S^{\Delta}_{\tau^R}-C_{\tau^R})d[H^R,\widehat{H}^R]_u \Bigg] \, . \nonumber
\end{align}

Notice that, since $t \in [0, \tau^R \wedge T]$, we have
\begin{align}
   \int_{]t,T]} B_u^{-1}H^R_u dD_u &= \int_{]t,\tau^R[}B_u^{-1}H^R_udD_u+ \int_{[\tau^R,T]}B_u^{-1}H^R_udD_u \nonumber \\[0.05in]
   & =  \int_{[\tau^R,T]}B_u^{-1} H^R_u dD_u \, . \label{eq:StauD}
\end{align}
Therefore,
\begin{align}
 &\1_{\{t \leq  \tau^R \wedge T\}} \bE_{\tau^R} \Big[ \int_{]{t},T]}B_u^{-1}H^R_udD_u \Big]
  =  \1_{\{t \leq  \tau^R \wedge T\}}\bE_{\tau^R} \Big[ \int_{[\tau^R,T]}B_u^{-1}H^R_udD_u \Big]\nonumber \\
   & = \1_{\{t \leq  \tau^R \wedge T\}}\1_{\{\tau^R \leq T\}} B^{-1}_{\tau^R}( S_{\tau^R} + \Delta D_{\tau^R} )
   = \1_{\{t \leq  \tau^R \}}\1_{\{\tau^R \leq T\}} B^{-1}_{\tau^R}( S_{\tau^R} + \Delta D_{\tau^R} )\,  .
\end{align}
Taking conditional expectation given $\cG_t$ and using the tower property in \eqref{eq:condTau} reads
\begin{align}
\1_{\{t \leq  \tau^R\wedge T \}}(S_t-{S}_t^R) &= \1_{\{t \leq  \tau^R \wedge T \}}B_t\bE_t\Big[\int_{]{t},T]}B_u^{-1}(dD_u-d{\cD}^R_u )\Big] \label{eq:condt}  \\[0.05in]
& \hspace{-3cm} =\1_{\{t \leq  \tau^R\wedge T \}} B_t\bE_t\Big[B^{-1}_{\tau^R}\big( \1_{\{\tau^R \leq T\}}(S^{\Delta}_{\tau^R} - C_{\tau^R})
- (R_1(S^{\Delta}_{\tau^R}-C_{\tau^R})^{+} -(S^{\Delta}_{\tau^R}-C_{\tau^R})^{-})\1_{\{\tau^R=\tau_1 \leq T\}} \nonumber \\[0.05in]
& \hspace{-2cm} +(R_2(S^{\Delta}_{\tau^R}-C_{\tau^R})^{-}-(S^{\Delta}_{\tau^R}-C_{\tau^R})^{+})\1_{\{ \tau^R=\tau_2 \leq T\}} \nonumber  \\[0.05in]
& \hspace{-2cm} +(S^{\Delta}_{\tau^R}-C_{\tau^R})\1_{\{\tau^R=\tau_1=\tau_2 \leq T\}}
-(S^{\Delta}_{\tau^R}-C_{\tau^R})\1_{\{\tau^R = \widehat{\tau}^R \leq T\}}\big)\Big] \, . \nonumber
\end{align}
Since
$$(S^Delta_{\tau^R}-C_{\tau^R})=(S^\Delta_{\tau^R}-C_{\tau^R})^{+}-(S^Delta_{\tau^R}-C_{\tau^R})^{-} \, ,$$
it follows that \eqref{eq:condt} is equivalent to
\begin{align*}
\1_{\{t \leq  \tau^R\wedge T \}}(S_t-{S}_t^R) &=\1_{\{t \leq  \tau^R\wedge T \}}B_t\bE_t\Big[B^{-1}_{\tau^R}\big(\1_{\{\tau^R \leq T\}}(S^{\Delta}_{\tau^R}-C_{\tau^R})  \\[0.05in]
& \quad  -(R_1(S^{\Delta}_{\tau^R}-C_{\tau^R})^{+}+(S^{\Delta}_{\tau^R}-C_{\tau^R})
-(S^{\Delta}_{\tau^R}-C_{\tau^R})^{+})\1_{\{\tau^R=\tau_1 \leq T\}}  \\[0.05in]
& \quad +(R_2(S^{\Delta}_{\tau^R}-C_{\tau^R})^{-}-(S^{\Delta}_{\tau^R}-C_{\tau^R})^{-}-(S^{\Delta}_{\tau^R}-C_{\tau^R}))\1_{\{ \tau^R=\tau_2\leq T\}}  \\[0.05in]
& \quad +(S^{\Delta}_{\tau^R}-C_{\tau^R})\1_{\{\tau^R=\tau_1=\tau_2 \leq T \}}
-(S^{\Delta}_{\tau^R}-C_{\tau^R})\1_{\{ \tau^R = \widehat{\tau}^R \leq T\}}\big)\Big] \, .
\end{align*}
After simplifying the terms above, we obtain
\begin{align*}
\1_{\{t \leq  \tau^R\wedge T \}}(S_t-{S}_t^R)  &=\1_{\{t \leq  \tau^R\wedge T \}}B_t\bE_t\Big[B^{-1}_{\tau^R}\big(\1_{\{\tau^R \leq T\}}(S^{\Delta}_{\tau^R}-C_{\tau^R}) \\[0.05in]
& \quad  +(1-R_1)(S^{\Delta}_{\tau^R}-C_{\tau^R})^{+}\1_{\{\tau^R=\tau_1 \leq T\}}
-(S^{\Delta}_{\tau^R}-C_{\tau^R})\1_{\{ \tau^R=\tau_1 \leq T\}}  \\[0.05in]
& \quad -(1-R_2)(S^{\Delta}_{\tau^R}-C_{\tau^R})^{-}\1_{\{ \tau^R=\tau_2 \leq T\}}
-(S^{\Delta}_{\tau^R}-C_{\tau^R})\1_{\{\tau^R=\tau_2 \leq T\}} \\[0.05in]
& \quad +(S^{\Delta}_{\tau^R}-C_{\tau^R})\1_{\{\tau^R=\tau_1=\tau_2 \leq T\}}
-(S^{\Delta}_{\tau^R}-C_{\tau^R})\1_{\{\tau^R = \widehat{\tau}^R \leq T\}}\big)\Big] \, ,
\end{align*}
which is equivalent to
\begin{align*}
\1_{\{t \leq  \tau^R\wedge T \}}(S_t-{S}_t^R)  &= \1_{\{t \leq  \tau^R\wedge T \}}B_t\bE_t \Big [B^{-1}_{\tau^R}\big[\1_{\{\tau^R \leq T\}}(S^{\Delta}_{\tau^R}-C_{\tau^R})-\1_{\{\tau^R \leq T\}}(S^{\Delta}_{\tau^R}-C_{\tau^R})  \\[0.05in]
&  +(1-R_1)(S^{\Delta}_{\tau^R}-C_{\tau^R})^{+}\1_{\{\tau^R=\tau_1\leq T\}}
 -(1-R_2)(S^{\Delta}_{\tau^R}-C_{\tau^R})^{-}\1_{\{\tau^R=\tau_2\leq T\}}\big] \Big] \, .
\end{align*}
Finally, we find that
\begin{align*}
S_t-{S}_t^R &= B_t \bE_t\Big[\1_{\{\tau^R=\tau_1\leq T\}}B_{\tau^R}^{-1}(1-R_1)(S^{\Delta}_{\tau^R}-C_{\tau^R})^{+}\Big]  \\[0.05in]
& \quad -B_t \bE_t\Big[\1_{\{ \tau^R=\tau_2\leq T\}}B_{\tau^R}^{-1}(1-R_2)(S^{\Delta}_{\tau^R}-C_{\tau^R})^{-}\Big] \, ,
\end{align*}
on the set $t \in [0,\tau^R \wedge T]$, which proves our claim.
\end{proof}
\begin{remark}
Note that since there are no losses associated with the trigger events other than defaults, and since CVA (as well as CVA$^R$) only reflects the expected losses, these cases do not appear directly in \eqref{CVAR-1}.
\end{remark}
Now, similar to \eqref{eq:DVA}, we can define
\begin{align*}
\text{UCVA}^R_t&:=B_t\bE_t[\1_{\{ \tau^R =\tau_1\leq T\}}B_{\tau^R }^{-1}(1-R_1)(S^{\Delta}_{\tau^R}-C_{\tau^R })^{+}] \, , \\[0.05in]
\text{DVA}^R_t &:=B_t\bE_t[\1_{\{ \tau^R=\tau_2\leq T\}}B_{\tau^R}^{-1}(1-R_2)(S^{\Delta}_{\tau^R}-C_{\tau^R})^{-}] \, ,
\end{align*}
for $t \in [0,\tau^R \wedge T]$. Therefore, the credit valuation adjustment representation found in \eqref{CVAR-1} can be decomposed as $\text{CVA}^R = \text{UCVA}^R - \text{DVA}^R$.

\begin{remark}
Note that although banks report on DVA (or DVA$^R$ in our case) in their earnings reports, it is not included in determining the capital levels.
This is also stated in \cite[Paragraph 75]{BISbasel3} as
\begin{quote}
Derecognise in the calculation of Common Equity Tier 1, all unrealised gains and
losses that have resulted from changes in the fair value of liabilities that are due to changes
in the bank's own credit risk.
\end{quote}
Therefore, Basel III framework does not allow the banks to account for DVA in their regulatory capital calculations (see also \cite{BISdva} for a detailed discussion).
The main reason of this treatment of DVA in Basel III is to not to allow banks to have the value of their liabilities decrease while their credit risk is increasing.
\end{remark}

It is important to observe the difference between CVA and CVA$^R$ processes, which indicates the change in the CVA due to rating triggers. This leads us to introduce the following concept,
\begin{definition}
The Rating Valuation Adjustment (RVA) process, is defined as
\begin{align}\label{RVA}
\text{RVA}_t &= \text{CVA}_t-\text{CVA}_t^R \, ,
\end{align}
for all $t \in [0,\tau^R \wedge T]$.
\end{definition}

The rating valuation adjustment term defined above has the following representation.
\begin{proposition}\label{proposition:RVA}
The  \textup{RVA} process can be represented as
\begin{align*}
\textup{RVA}_t  &=B_t\bE_t[\1_{\{\tau^R<\tau =\tau_1\leq T\}}B_{\tau^R}^{-1}(1-R_1)(S^{\Delta}_{\tau}-C_{\tau})^{+}]  \\[0.05in]
& \quad -B_t\bE_t[\1_{\{ \tau^R<\tau =\tau_2\leq T\}}B_{\tau^R}^{-1}(1-R_2)(S^{\Delta}_{\tau}-C_{\tau})^{-}] \, ,
\end{align*}
for all $t \in [0,\tau^R \wedge T]$.
\end{proposition}
\begin{proof}
From \eqref{CVA} and \eqref{CVAR-1} we obtain
\begin{align*}
\text{CVA}_t-\text{CVA}_t^R &= B_t\bE_t\Big[\1_{\{ \tau =\tau_1\leq T\}}B_{\tau}^{-1}(1-R_1)(S^{\Delta}_{\tau}-C_{\tau })^{+}\Big] \\[0.05in]
& \quad - B_t\bE_t\Big[\1_{\{ \tau =\tau_2\leq T\}}B_{\tau}^{-1}(1-R_2)(S^{\Delta}_{\tau}-C_{\tau })^{-}\Big]  \\[0.05in]
& \quad - B_t \bE_t\Big[\1_{\{ \tau^R=\tau_1\leq T\}}B_{\tau^R}^{-1}(1-R_1)(S^{\Delta}_{\tau^R}-C_{\tau^R})^{+}\Big]  \\[0.05in]
& \quad + B_t \bE_t\Big[\1_{\{ \tau^R=\tau_2\leq T\}}B_{\tau^R}^{-1}(1-R_2)(S^{\Delta}_{\tau^R}-C_{\tau^R})^{-}\Big] \, ,
\end{align*}
which can be written as,
\begin{align*}
\text{CVA}_t-\text{CVA}_t^R &= B_t\bE_t\Big[\1_{\{ \tau =\tau_1\leq T\}}B_{\tau_1}^{-1}(1-R_1)(S^{\Delta}_{\tau_1}-C_{\tau_1})^{+}\Big]  \\[0.05in]
& \quad - B_t \bE_t\Big[\1_{\{ \tau = \tau_2\leq T\}}B_{\tau_2}^{-1}(1-R_2)(S^{\Delta}_{\tau_2}-C_{\tau_2})^{-}\Big]  \\[0.05in]
& \quad - B_t \bE_t\Big[\1_{\{ \tau^R=\tau_1\leq T\}}B_{\tau_1}^{-1}(1-R_1)(S^{\Delta}_{\tau_1}-C_{\tau_1})^{+}\Big]  \\[0.05in]
& \quad + B_t \bE_t\Big[\1_{\{ \tau^R=\tau_2\leq T\}}B_{\tau_2}^{-1}(1-R_2)(S^{\Delta}_{\tau_2}-C_{\tau_2})^{-}\Big] \, .
\end{align*}
Therefore, simplifying the terms above yields
\begin{align*}
\text{CVA}_t-\text{CVA}_t^R &=
B_t\bE_t\Big[(\1_{\{ \tau =\tau_1\leq T\}}-\1_{\{ \tau^R = \tau_1\leq T\}})
B_{\tau_1}^{-1}(1-R_1)(S^{\Delta}_{\tau_1}-C_{\tau_1})^{+}\Big]  \\[0.05in]
& \quad - B_t \bE_t\Big[(\1_{\{ \tau = \tau_2\leq T\}} - \1_{\{ \tau^R = \tau_2\leq T\}}) B_{\tau_2}^{-1}(1-R_2)(S^{\Delta}_{\tau_2}-C_{\tau_2})^{-}\Big] \, ,
\end{align*}
which is equivalent to
\begin{align*}
\text{CVA}_t-\text{CVA}_t^R &= B_t\bE_t\Big[\1_{\{ \tau^R < \tau =\tau_1\leq T\}}B_{\tau_1}^{-1}(1-R_1)(S^{\Delta}_{\tau_1}-C_{\tau_1})^{+}\Big]  \\[0.05in]
& \quad - B_t \bE_t\Big[\1_{\{\tau^R < \tau = \tau_2\leq T\}}B_{\tau_2}^{-1}(1-R_2)(S^{\Delta}_{\tau_2}-C_{\tau_2})^{-}\Big] \, ,
\end{align*}
for $t \in [0, \tau^R \wedge T]$, which proves the result in view of \eqref{RVA}.
\end{proof}

\begin{remark}
  Note that RVA can be positive or negative. If RVA is positive then there is a decrease in the bilateral CVA.  If RVA is negative then this indicates an increase in the bilateral CVA due to adding rating triggers. Furthermore, RVA is always non-negative in case of measuring unilateral counterparty risk ($\tau_2 = \infty$).
\end{remark}
Let us define
\begin{align*}
  \text{URVA}_t:& = B_t\bE_t[\1_{\{\tau^R<\tau =\tau_1\leq T\}}B_{\tau^R}^{-1}(1-R_1)(S^{\Delta}_{\tau}-C_{\tau})^{+}] \, ,  \\[0.05in]
  \text{DRVA}_t:& = B_t\bE_t[\1_{\{ \tau^R<\tau =\tau_2\leq T\}}B_{\tau^R}^{-1}(1-R_2)(S^{\Delta}_{\tau}-C_{\tau})^{-}] \, .
\end{align*}
for $t \in [0,\tau^R \wedge T]$. Therefore, RVA has the following decomposition,
\begin{align*}
\text{RVA}_t &= \text{URVA}_t - \text{DRVA}_t \, ,
\end{align*}
for $t \in [0,\tau^R \wedge T]$.
Here URVA represents the expected loss if the counterparty defaults first which is preceded by a rating trigger.
Similarly, DRVA is the expected loss in case the investor defaults first after a rating trigger.
Therefore, including rating triggers provision in an OTC contract provides protection from losses due to default events which happen after a credit downgrade. Accordingly, the value of the contract is adjusted for this protection, as shown in the following

\begin{corollary}
We have the following decomposition for the counterparty risk-free price process
\begin{align*}
S_t &=   S_t^R + \text{CVA}^R_t \\
&=S_t^R + \text{CVA}_t- \text{RVA}_t  \\
&=S_t^R + \text{UCVA}_t - \text{DVA}_t - \text{RVA}_t \\
& = S_t^R + \text{UCVA}_t - \text{DVA}_t - \text{URVA}_t + \text{DRVA}_t \, ,
\end{align*}
for $t \in [0,\tau^R \wedge T ]$.
\end{corollary}

\subsection{Dynamic Collateralization}\label{section:collateral}

In bilateral margin agreements, counterparties are required to post collateral as soon as the clean price of the contract exceeds thresholds, which are defined in CSA (see \cite{ISDAcsa94}).
In particular, these thresholds are defined in terms of the credit ratings of the counterparties.
Specifically, the collateral threshold of a counterparty decreases as a result of a credit rating downgrade and increases as a result of a credit rating upgrade. Consequently, a counterparty with higher credit rating will have higher threshold than a counterparty with a lower credit rating.

It is important to note that there is an adverse relation between the margin requirements and the credit ratings.
A credit downgrade along with higher borrowing rates and exposures forces the companies to post increasing amounts of collateral to their counterparties, which can be fatal.
For example, the ratings linked collateral thresholds, coupled with rehypothecation, have been considered to be one of the key drivers of AIG's collapse in 2008.
Before 2007, as a `AAA' rated company, AIG had not been required to post any collateral for most of its derivatives transactions.
However, after several downgrades AIG had posted more than \$40 billion in collateral as of November 2008 (see \cite{ISDAaig2009} for details).

Thus, one of the key issues in modeling of the collateral process\footnote{Since in this paper we only consider symmetric cash flows (form the point of view of both the parties), we only need to model a single collateral process.} is the issue of modeling of the thresholds. In what follows, we shall model the collateral threshold for the counterparty at time $t$, say $\Gamma^1_t$,  as $\Gamma^1_t=\gamma^1(t,X^1_t,S_t)$, where $\gamma^1\, : \, [0,T]\times \cK \times \bR \rightarrow \bR_+$ is a measurable function. Similarly, we shall model the collateral threshold for the investor at time $t$, say $\Gamma^2_t$,  as $\Gamma^2_t=\gamma^2(t,X^2_t,S_t)$, where $\gamma^2\, : \, [0,T]\times \cK \times \bR \rightarrow \bR_-$ is a measurable function.

For a proper modeling of the collateral we need to consider the so called \emph{independent amounts} (i.e. initial margins) posted by the counterparty and the investor by the constants $\beta_1 \in \bR_+$ and $\beta_2 \in \bR_-$, respectively. We also need to consider the so called \emph{minimum transfer amount} (MTA), which is a positive constant denoted by $\theta$ and the \emph{margin period of risk}, which is again a positive constant denoted by $\Delta$.\footnote{We refer the reader to Bielecki et al. \cite{BCI2011,Bielecki2011a} for a detailed discussion and definitions.}

According to the standard industry practice collateral amounts are adjusted at fixed tenor dates, termed {\it margin call dates}. Let us denote the margin call dates by $0<t_1 < \ldots < t_n< T$.
On the margin call date $t_i$, if the exposure is above the counterparty's current threshold, $\Gamma^1_{t_i}$, and if the difference between the current exposure and the collateral amount is greater than the MTA the counterparty posts collateral and updates the margin account; otherwise, no collateral exchange takes place since the transfer amount is less than the MTA.
Likewise, the investor delivers collateral on the margin call date $t_i$, if the exposure is below investor's threshold, $\Gamma^2_{t_i}$, and the difference between the current exposure and the collateral amount is greater than MTA (cf. \cite{ISDA2005}, pages 52--56).
Note that under these covenants a collateral transfer is allowed only if it is greater than the MTA amount.

In accordance with the above discussion the collateral process is modeled as follows in this paper:

We set $C_0=0$. Then, for $i=1,2,\ldots,n$, we define
\begin{align*}
C_{t} &:= \1_{\{S_{t_i}+B_{t_i}(\beta_1-\beta_2)-\Gamma^1_{t_i}-C_{t_i}>\theta\}}(S_{t_i}+B_{t_i}(\beta_1-\beta_2)-\Gamma^1_{t_i}-C_{t_i}) \\[0.05in]
& \quad +\1_{\{S_{t_i}+B_{t_i}(\beta_2-\beta_1)-\Gamma^2_{t_i}-C_{t_i}<-\theta\} }( S_{t_i}+B_{t_i}(\beta_2-\beta_1)-\Gamma^2_{t_i}-C_{t_i}) + C_{t_i},
\end{align*}
for $t\in (t_i,t_{i+1}]$, on the set $\{t_i<\tau^R \}$.
Moreover, $C_t = C_{\tau^R}$ on the set $\{ \tau^R \leq t\leq  \tau^R + \Delta \}$, where $\Delta$ represents the margin period of risk.

Observe that the collateral increments at each margin call date $t_i<\tau^R$ can now be represented as,
\begin{align*}
\Delta C_{t_i}:&=C_{t_i+} - C_{t_i}   \\[0.05in]
&= \1_{\{S_{t_i}+B_{t_i}(\beta_1-\beta_2)-\Gamma^1_{t_i}-C_{t_i}>\theta\}}(S_{t_i}+B_{t_i}(\beta_1-\beta_2)-\Gamma^1_{t_i}-C_{t_i})   \\[0.05in]
& + \1_{\{S_{t_i}+B_{t_i}(\beta_2-\beta_1)-\Gamma^2_{t_i}-C_{t_i}<-\theta\} }( S_{t_i}+B_{t_i}(\beta_2-\beta_1)-\Gamma^2_{t_i}-C_{t_i}) \, .
\end{align*}

 In Section \ref{section:Apps} we assume, for simplicity, that the margin period of risk, independent amounts and minimum transfer amount to be zero. Thus, the collateral amount at time $t$ (from the point of view of the investor) is given as
\begin{align*}
C_t & = \1_{\{S_{t_i}-\Gamma^1_{t_i}>C_{t_i} \}}(S_{t_i}-\Gamma^1_{t_i}-C_{t_i})+\1_{\{S_{t_i}-\Gamma^2_{t_i}< C_{t_i}\} }( S_{t_i}-\Gamma^2_{t_i}-C_{t_i}) + C_{t_i},
\end{align*}
for $t \in (t_i,t_{i+1}]$.
Furthermore, we consider the following structure for the collateral thresholds
\begin{equation*}
 \gamma^i (t,x,s) = \rho^i(t,x) s,\quad \quad i=1,2,
\end{equation*}
where $\rho^i \, : \, [0,T]\times \cK \rightarrow [0,1]$ is a measurable function.
The functions $\rho^1$ and $\rho^2$ represent the \emph{collateral rates} for the counterparty and the investor at time $t$, respectively .
Essentially, the collateral rates indicate the percentage of exposure at time $t$.

We introduce two specifications of collateral rates:\footnote{Recall that the credit ratings of each credit name take values in the set $\cK=\{1,2,\ldots ,K\}$, where  $K$ represents the default and where $1$ represents the highest possible rating. }

\begin{itemize}
\item
The \emph{linear case}:
\begin{equation*}
\rho^i_l(t,x) :=\frac{K-x}{K-1}
\end{equation*}
for all $i=1,2$. In particular, $\rho^i(t,1) =1$ and  $\rho^i(t,K) =0.$

\item
The \emph{exponential case}:
\begin{equation*}
\rho^i_e(t,x):=
\begin{cases}
e^{1-x} & \text{if $x<K$}  \\[0.05in]
0 & \text{if $x=K$}
\end{cases}
\end{equation*}
for all $i=1,2$.
\end{itemize}

Note that in practice, the threshold levels are set in CSA documents for different rating levels. However, here we propose two alternative methods for determining the collateral threshold levels. In the linear case, collateral thresholds increase or decrease linearly with the credit qualities of the counterparties. Similarly, in the exponential case, the collateral thresholds exponentially increase or decrease with the credit ratings. Therefore, the collateral rates in the exponential case are always less than the ones in the linear case, which leads to lower collateral thresholds, and as a result more collateral being kept in the margin account.

\subsection{Rehypothecation Risk}\label{section:rehypothecation}

We now consider the case that the collateral receiver (counterparty or the investor) can rehypothecate the collateral.
Rehypothecation refers to the case that the collateral receiver is able to freely use the collateral in the margin account for investment and funding purposes.

Let us define a $\cG_{\tau_1}$-measurable random variable $R^h_1$ and a $\cG_{\tau_2}$-measurable random variable $R^h_2$ as the recovery rates of the rehypothecated collateral for the investor and the counterparty.
Following Brigo et al. \cite{BCPP2011}, we assume that $R_1 \leq R^h_1$ and $R_2 \leq R^h_2$, since in case of a default the collateral has priority among other liabilities.

Let us now define the cumulative dividend process associated with the counterparty risky contract with rehypothecation.
\begin{definition}\label{def:cumDivRatingRehyp}
Cumulative dividend process of a counterparty risky contract that takes rehypothecation risk into account is represented as,
\begin{align*}
{\cD}^{R,h}_t &=(1-H^R_t)D_t
+D_{\tau^R-}H^R_t + \1_{\{\tau^R \leq T\}}\Big(  \widetilde{C}_{\tau^R}H^R_t \\[0.05in]
& \quad +\left(R_1(S^{\Delta}_{\tau^R} - \widetilde{C}_{\tau^R})^{+}-(S^{\Delta}_{\tau^R} -\widetilde{C}_{\tau^R})^{-}\right)[H^R,H^1]_t  \\[0.05in]
& \quad -\left(R_2(S^{\Delta}_{\tau^R} - \widetilde{C}_{\tau^R})^{-}-(S^{\Delta}_{\tau^R} -\widetilde{C}_{\tau^R})^{+}\right)[H^R,H^2]_t \\[0.05in]
& \quad -\left(S^{\Delta}_{\tau^R} - \widetilde{C}_{\tau^R}\right)[[H^R,H^1] \, ,H^2]_t
+\left(S^{\Delta}_{\tau^R} - \widetilde{C}_{\tau^R}\right)[H^R,\widehat{H}^{R}]_t \Big) \, ,
\end{align*}
for all $t \in [0,T]$,
where
\begin{align*}
  \widetilde{C}_{\tau^R} &= C_{\tau^R}\Big[
\1_{{\tau^R}=\tau_1 \neq \tau_{2}}(R^h_1 \1_{C_\tau^R>0}+\1_{C_{\tau^R} \leq 0}) + \1_{{\tau^R}=\tau_{2} \neq \tau_1}(\1_{C_{\tau^R} > 0}+R^h_{2}\1_{C_{\tau^R}\leq0}) \\
& \quad + \1_{{\tau^R}=\tau_{1}=\tau_2}(R^h_1\1_{C_{\tau^R} >0}+R^h_{2}\1_{C_{\tau^R} \leq 0}) +
\1_{{\tau^R}=\widehat{\tau}^R}\Big].
\end{align*}
\end{definition}

We are now ready to define the ex-dividend price processes associated with a counterparty risky contract with rating triggers and rehypothecation risk.

\begin{definition}\label{def:CDSwRTrhyp}
The ex-dividend price process ${S}^{R,h}$ of a counterparty risky contract maturing at time T, with rating triggers and rehypothecation risk is defined as,
\begin{equation*}
{S}_t^{R,h}=B_t\bE_t \Big[ \int_{]t,T]} B_u^{-1} d {\cD}^{R,h}_u \Big] \, ,
\end{equation*}
for all $t\in[0,T]$.
\end{definition}

Next, we give the definition of credit valuation adjustment of a contract with rating triggers in presence of rehypothecation risk.
\begin{definition}
The credit valuation adjustment with rating triggers taking the rehypothecation risk into account is defined as,
\begin{equation}
\text{CVA}_t^{R,h}=S_t-{S}_t^{R,h} \, , \label{eq:CVArehyp}
\end{equation}
for all $t\in[0,\tau^R \wedge T]$.
\end{definition}

This form of the counterparty-risky cumulative dividend process leads to the following representation for the bilateral CVA.
\begin{proposition}
The bilateral Credit Valuation Adjustment process with rehypothecation risk defined in \eqref{eq:CVArehyp} can be represented as
  \begin{align}
\text{CVA}_t^{R,h} &=B_t\bE_t\Big[\1_{\{ \tau^R=\tau_1 \leq T\}} B^{-1}_{\tau^R}
 (1-R_1)(S^\Delta_{\tau^R}-\widetilde{C}^1_{\tau^R})^{+}\Big] \nonumber  \\[0.05in]
& \quad - B_t\bE_t\Big[\1_{\{ \tau^R=\tau_2 \leq T\}}B^{-1}_{\tau^R} (1-R_2)(S^\Delta_{\tau^R}-\widetilde{C}^2_{\tau^R})^{-} \Big] \label{eq:CVARh} \, ,
\end{align}
for all $t\in[0,\tau^R \wedge T]$, where
\begin{align*}
  \widetilde{C}^1_{\tau^R} &= C_{\tau^R}\Big[
\1_{{\tau^R}=\tau_1 \neq \tau_{2}}(R^h_1 \1_{C_\tau^R>0}+\1_{C_{\tau^R} \leq 0})
 + \1_{{\tau^R}=\tau_{1}=\tau_2}(R^h_1\1_{C_{\tau^R} >0}+R^h_{2}\1_{C_{\tau^R} \leq 0})\Big] \, ,
\end{align*}
and
\begin{align*}
  \widetilde{C}^2_{\tau^R} &= C_{\tau^R}\Big[
 \1_{{\tau^R}=\tau_{2} \neq \tau_1}(\1_{C_{\tau^R} > 0}+R^h_{2}\1_{C_{\tau^R}\leq0})
 + \1_{{\tau^R}=\tau_{1}=\tau_2}(R^h_1\1_{C_{\tau^R} >0}+R^h_{2}\1_{C_{\tau^R} \leq 0}) \Big].
\end{align*}
\end{proposition}

\begin{proof}
Using Definition \ref{def:cumDivRatingRehyp}, we have
\begin{align*}
  dD_t-d{\cD}^{R,h} _t &= dD_t- (1-H^R_t)dD_t -D_{t-}dH^R_t +D_{\tau^R-}dH^R_t - \1_{\{\tau^R \leq T\}}\widetilde{C}_{\tau^R}dH^R_t  \\[0.05in]
& \quad -\1_{\{\tau^R \leq T\}}\big(R_1(S^{\Delta}_{\tau^R}-\widetilde{C}_{\tau^R})^{+}-(S^{\Delta}_{\tau^R}-\widetilde{C}_{\tau^R})^{-}\big)d[H^R,H^1]_t \\[0.05in]
& \quad +\1_{\{\tau^R \leq T\}}\big(R_2(S^{\Delta}_{\tau^R}-\widetilde{C}_{\tau^R})^{-} -(S^{\Delta}_{\tau^R}-\widetilde{C}_{\tau^R})^{+}\big)d[H^R,H^2]_t \\[0.05in]
& \quad +\1_{\{\tau^R \leq T\}}(S^{\Delta}_{\tau^R}-\widetilde{C}_{\tau^R})d[[H^R,H^1] \, ,H^2]_t  \\[0.05in]
& \quad -\1_{\{\tau^R \leq T\}}(S^{\Delta}_{\tau^R}-\widetilde{C}_{\tau^R})d[H^R,\widetilde{H}^R]_t \, ,
\end{align*}
Integrating both sides leads to,
\begin{align*}
  \int_{]{t},T]}B_u^{-1}(dD_u-d{\cD}^{R,h} _u) & = \int_{]{t},T]}B_u^{-1}H^R_udD_u - \int_{]{t},T]}B_u^{-1}D_{u-}dH^R_u + \int_{]{t},T]}B_u^{-1}D_{\tau^R-}dH^R_u \\[0.05in]
& \hspace{-2cm} -\int_{]{t},T]}\1_{\{\tau^R \leq T\}}B_u^{-1}\big(R_1(S^{\Delta}_{\tau^R}-\widetilde{C}_{\tau^R})^{+}  -(S^{\Delta}_{\tau^R}-\widetilde{C}_{\tau^R})^{-}\big)d[H^R,H^1]_u \\[0.05in]
& \hspace{-2cm} +\int_{]{t},T]}\1_{\{\tau^R \leq T\}}B_u^{-1}\big(R_2(S^{\Delta}_{\tau^R}-\widetilde{C}_{\tau^R})^{-} -(S^{\Delta}_{\tau^R}-\widetilde{C}_{\tau^R})^{+}\big)d[H^R,H^2]_u \\[0.05in]
& \hspace{-2cm} +\int_{]{t},T]}\1_{\{\tau^R \leq T\}}B_u^{-1}(S^{\Delta}_{\tau^R}-\widetilde{C}_{\tau^R})d[[H^R,H^1] \, ,H^2]_u  \\[0.05in]
& \hspace{-2cm} -\int_{]{t},T]}\1_{\{\tau^R \leq T\}}B_u^{-1}(S^{\Delta}_{\tau^R}-\widetilde{C}_{\tau^R})d[H^R,\widetilde{H}^R]_u- \int_{]{t},T]}\1_{\{\tau^R \leq T\}}B_u^{-1}\widetilde{C}_{\tau^R}dH^R_u \, .
\end{align*}
Since,
\begin{equation*}
   \int_{]{t},T]}B_u^{-1}D_{\tau^R-}dH^R_u - \int_{]{t},T]}B_u^{-1}D_{u-}dH^R_u = 0 \, ,
\end{equation*}
we obtain,
\begin{align*}
  \int_{]{t},T]}B_u^{-1}(dD_u-d{\cD}^{R,h} _u) & = \int_{]{t},T]}B_u^{-1}H^R_udD_u - \int_{]{t},T]}\1_{\{\tau^R \leq T\}}B_u^{-2}\widetilde{C}_{\tau^R}dH^R_u \\[0.05in]
& \hspace{-2cm} -\int_{]{t},T]}\1_{\{\tau^R \leq T\}}B_u^{-1}\big(R_1(S^{\Delta}_{\tau^R}-\widetilde{C}_{\tau^R})^{+} -(S^{\Delta}_{\tau^R}-\widetilde{C}_{\tau^R})^{-}\big)d[H^R,H^1]_u \\[0.05in]
& \hspace{-2cm} +\int_{]{t},T]}\1_{\{\tau^R \leq T\}}B_u^{-1}\big(R_2(S^{\Delta}_{\tau^R}-\widetilde{C}_{\tau^R})^{-}  -(S^{\Delta}_{\tau^R}-\widetilde{C}_{\tau^R})^{+}\big)d[H^R,H^2]_u \\[0.05in]
& \hspace{-2cm} +\int_{]{t},T]}\1_{\{\tau^R \leq T\}}B_u^{-1}(S^{\Delta}_{\tau^R}-\widetilde{C}_{\tau^R})d[[H^R,H^1] \, ,H^2]_u  \\[0.05in]
& \hspace{-2cm} -\int_{]{t},T]}\1_{\{\tau^R \leq T\}}B_u^{-1}(S^{\Delta}_{\tau^R}-\widetilde{C}_{\tau^R})d[H^R,\widetilde{H}^R]_u \, .
\end{align*}
Conditioning on $\tau^R$, we get
\begin{align}
 \1_{\{t \leq  \tau^R \wedge T\}}\ \bE_{\tau^R} \Big[ \int_{]{t},T]}B_u^{-1}(dD_u-d{\cD}^{R,h} _u) \Big]
 &  =\1_{\{t \leq  \tau^R \wedge T\}}\ \bE_{\tau^R} \Bigg[ \int_{]{t},T]}B_u^{-1}H^R_udD_u \nonumber \\[0.05in]
  & \hspace{-4cm} - \int_{]{t},T]}\1_{\{\tau^R \leq T\}}B_u^{-1}\widetilde{C}_{\tau^R}dH^R_u \nonumber \\[0.05in]
& \hspace{-4cm} -\int_{]{t},T]}\1_{\{\tau^R \leq T\}}B_u^{-1}\big(R_1(S^{\Delta}_{\tau^R}-\widetilde{C}_{\tau^R})^{+}  -(S^{\Delta}_{\tau^R}-\widetilde{C}_{\tau^R})^{-}\big)d[H^R,H^1]_u \nonumber \\[0.05in]
& \hspace{-4cm} +\int_{]{t},T]}\1_{\{\tau^R \leq T\}}B_u^{-1}\big(R_2(S^{\Delta}_{\tau^R}-\widetilde{C}_{\tau^R})^{-} -(S^{\Delta}_{\tau^R}-\widetilde{C}_{\tau^R})^{+}\big)d[H^R,H^2]_u \nonumber \\[0.05in]
& \hspace{-4cm} +\int_{]{t},T]}\1_{\{\tau^R \leq T\}}B_u^{-1}(S^{\Delta}_{\tau^R}-\widetilde{C}_{\tau^R})d[[H^R,H^1] \, ,H^2]_u \label{eq:condTau} \\[0.05in]
& \hspace{-4cm} -\int_{]{t},T]}\1_{\{\tau^R \leq T\}}B_u^{-1}(S^{\Delta}_{\tau^R}-\widetilde{C}_{\tau^R})d[H^R,\widetilde{H}^R]_u \Bigg] \, . \nonumber
\end{align}

Notice that, since $t \in [0, \tau^R \wedge T]$, we have
\begin{align}
   \int_{]t,T]} B_u^{-1}H^R_u dD_u &= \int_{]t,\tau^R[}B_u^{-1}H^R_udD_u+ \int_{[\tau^R,T]}B_u^{-1}H^R_udD_u \nonumber \\[0.05in]
   & =  \int_{[\tau^R,T]}B_u^{-1} H^R_u dD_u \, . \label{eq:StauD}
\end{align}
Therefore,
\begin{align}
 &\1_{\{t \leq  \tau^R \wedge T\}} \bE_{\tau^R} \Big[ \int_{]{t},T]}B_u^{-1}H^R_udD_u \Big]
  =  \1_{\{t \leq  \tau^R \wedge T\}}\bE_{\tau^R} \Big[ \int_{[\tau^R,T]}B_u^{-1}H^R_udD_u \Big]\nonumber \\
   & = \1_{\{t \leq  \tau^R \wedge T\}}\1_{\{\tau^R \leq T\}} B^{-1}_{\tau^R}( S_{\tau^R} + \Delta D_{\tau^R} )
   = \1_{\{t \leq  \tau^R \}}\1_{\{\tau^R \leq T\}} B^{-1}_{\tau^R}( S_{\tau^R} + \Delta D_{\tau^R} )\,  .
\end{align}
Taking conditional expectation given $\cG_t$ and using the tower property in \eqref{eq:condTau} reads
\begin{align}
\1_{\{t \leq  \tau^R\wedge T \}}(S_t-{S}_t^R) &= \1_{\{t \leq  \tau^R \wedge T \}}B_t\bE_t\Big[\int_{]{t},T]}B_u^{-1}(dD_u-d{\cD}^{R,h} _u )\Big] \label{eq:condt}  \\[0.05in]
& \hspace{-3cm} =\1_{\{t \leq  \tau^R\wedge T \}} B_t\bE_t\Big[B^{-1}_{\tau^R}\big( \1_{\{\tau^R \leq T\}}(S^{\Delta}_{\tau^R} - \widetilde{C}_{\tau^R}) \nonumber \\
& \hspace{-3cm} - (R_1(S^{\Delta}_{\tau^R}-\widetilde{C}_{\tau^R})^{+} -(S^{\Delta}_{\tau^R}-\widetilde{C}_{\tau^R})^{-})\1_{\{\tau^R=\tau_1 \leq T\}} \nonumber \\[0.05in]
& \hspace{-2cm} +(R_2(S^{\Delta}_{\tau^R}-\widetilde{C}_{\tau^R})^{-}-(S^{\Delta}_{\tau^R}-\widetilde{C}_{\tau^R})^{+})\1_{\{ \tau^R=\tau_2 \leq T\}} \nonumber  \\[0.05in]
& \hspace{-2cm} +(S^{\Delta}_{\tau^R}-\widetilde{C}_{\tau^R})\1_{\{\tau^R=\tau_1=\tau_2 \leq T\}}
-(S^{\Delta}_{\tau^R}-\widetilde{C}_{\tau^R})\1_{\{\tau^R = \widetilde{\tau}^R \leq T\}}\big)\Big] \, . \nonumber
\end{align}
Since
$$(S^\Delta_{\tau^R}-\widetilde{C}_{\tau^R})=(S^\Delta_{\tau^R}-\widetilde{C}_{\tau^R})^{+}-(S^\Delta_{\tau^R}-\widetilde{C}_{\tau^R})^{-} \, ,$$
it follows that \eqref{eq:condt} is equivalent to
\begin{align*}
\1_{\{t \leq  \tau^R\wedge T \}}(S_t-{S}_t^R) &=\1_{\{t \leq  \tau^R\wedge T \}}B_t\bE_t\Big[B^{-1}_{\tau^R}\big(\1_{\{\tau^R \leq T\}}(S^{\Delta}_{\tau^R}-\widetilde{C}_{\tau^R})  \\[0.05in]
& \quad  -(R_1(S^{\Delta}_{\tau^R}-\widetilde{C}_{\tau^R})^{+}+(S^{\Delta}_{\tau^R}-\widetilde{C}_{\tau^R})
-(S^{\Delta}_{\tau^R}-\widetilde{C}_{\tau^R})^{+})\1_{\{\tau^R=\tau_1 \leq T\}}  \\[0.05in]
& \quad +(R_2(S^{\Delta}_{\tau^R}-\widetilde{C}_{\tau^R})^{-}-(S^{\Delta}_{\tau^R}-\widetilde{C}_{\tau^R})^{-}-(S^{\Delta}_{\tau^R}-\widetilde{C}_{\tau^R}))\1_{\{ \tau^R=\tau_2\leq T\}}  \\[0.05in]
& \quad +(S^{\Delta}_{\tau^R}-\widetilde{C}_{\tau^R})\1_{\{\tau^R=\tau_1=\tau_2 \leq T \}}
-(S^{\Delta}_{\tau^R}-\widetilde{C}_{\tau^R})\1_{\{ \tau^R = \widetilde{\tau}^R \leq T\}}\big)\Big] \, .
\end{align*}
After simplifying the terms above, we obtain
\begin{align*}
\1_{\{t \leq  \tau^R\wedge T \}}(S_t-{S}_t^R)  &=\1_{\{t \leq  \tau^R\wedge T \}}B_t\bE_t\Big[B^{-1}_{\tau^R}\big(\1_{\{\tau^R \leq T\}}(S^{\Delta}_{\tau^R}-\widetilde{C}_{\tau^R}) \\[0.05in]
& \quad  +(1-R_1)(S^{\Delta}_{\tau^R}-\widetilde{C}_{\tau^R})^{+}\1_{\{\tau^R=\tau_1 \leq T\}}
-(S^{\Delta}_{\tau^R}-\widetilde{C}_{\tau^R})\1_{\{ \tau^R=\tau_1 \leq T\}}  \\[0.05in]
& \quad -(1-R_2)(S^{\Delta}_{\tau^R}-\widetilde{C}_{\tau^R})^{-}\1_{\{ \tau^R=\tau_2 \leq T\}}
-(S^{\Delta}_{\tau^R}-\widetilde{C}_{\tau^R})\1_{\{\tau^R=\tau_2 \leq T\}} \\[0.05in]
& \quad +(S^{\Delta}_{\tau^R}-\widetilde{C}_{\tau^R})\1_{\{\tau^R=\tau_1=\tau_2 \leq T\}}
-(S^{\Delta}_{\tau^R}-\widetilde{C}_{\tau^R})\1_{\{\tau^R = \widetilde{\tau}^R \leq T\}}\big)\Big] \, ,
\end{align*}
which is equivalent to
\begin{align*}
\1_{\{t \leq  \tau^R\wedge T \}}(S_t-{S}_t^R)  &= \1_{\{t \leq  \tau^R\wedge T \}}B_t\bE_t \Big [B^{-1}_{\tau^R}\big[\1_{\{\tau^R \leq T\}}(S^{\Delta}_{\tau^R}-\widetilde{C}_{\tau^R})-\1_{\{\tau^R \leq T\}}(S^{\Delta}_{\tau^R}-\widetilde{C}_{\tau^R})  \\[0.05in]
& \hspace{-1cm} +(1-R_1)(S^{\Delta}_{\tau^R}-\widetilde{C}_{\tau^R})^{+}\1_{\{\tau^R=\tau_1\leq T\}}
 -(1-R_2)(S^{\Delta}_{\tau^R}-\widetilde{C}_{\tau^R})^{-}\1_{\{\tau^R=\tau_2\leq T\}}\big] \Big] \, .
\end{align*}
Finally, we find that
\begin{align*}
S_t-{S}_t^R &= B_t \bE_t\Big[\1_{\{\tau^R=\tau_1\leq T\}}B_{\tau^R}^{-1}(1-R_1)(S^{\Delta}_{\tau^R}-\widetilde{C}_{\tau^R})^{+}\Big]  \\[0.05in]
& \quad -B_t \bE_t\Big[\1_{\{ \tau^R=\tau_2\leq T\}}B_{\tau^R}^{-1}(1-R_2)(S^{\Delta}_{\tau^R}-\widetilde{C}_{\tau^R})^{-}\Big] \, ,
\end{align*}
on the set $t \in [0,\tau^R \wedge T]$, which proves our claim.
\end{proof}

\begin{remark}
Observe that if $R^h_1 =R^h_2 = 1$ then $\text{CVA}^{R,h} = \text{CVA}^R$.
\end{remark}

Next, we consider the rating valuation adjustment in the presence of rehypothecation risk.
\begin{definition}
The Rating Valuation Adjustment process ($\text{RVA}^h$) with rehypothecation risk is defined as
\begin{align*}
\text{RVA}^h_t &= \text{CVA}_t-\text{CVA}_t^{R,h} \, ,
\end{align*}
for $t \in [0, \tau^R \wedge T]$.
\end{definition}

We have the following representation for RVA$^{R,h}$.
\begin{lemma}
RVA$^{R,h}$ can represented as
\begin{align*}
\text{RVA}^h_t &= \text{RVA}_t \\[0.05in]
& \hspace{-2cm}+ B_t\bE_t\Big[\1_{\{ \tau^R =\tau_1 \neq \tau_{2}\leq T\}}B_{\tau_1}^{-1}(1-R_1)
[\1_{C_{\tau_1} > 0}((S^{\Delta}_{\tau^R}-C_{\tau^R})^{+}-(S^\Delta_{\tau_1}-R^h_1C_{\tau_1})^+)]\\[0.05in]
& \hspace{-1cm} +\1_{\{ \tau^R =\tau_1 = \tau_{2}\leq T\}}B_{\tau_1}^{-1}(1-R_1)
[\1_{C_{\tau_1} > 0}((S^{\Delta}_{\tau^R}-C_{\tau^R})^{+}-(S^\Delta_{\tau_1}-R^h_1C_{\tau_1})^+) \\[0.05in]
& \hspace{-0.5cm} +\1_{C_{\tau_1} < 0}((S^{\Delta}_{\tau^R}-C_{\tau^R})^{+}-(S^\Delta_{\tau_1}-R^h_{2}C_{\tau_1})^{+})]\Big]  \\[0.05in]
& \hspace{-2cm} + B_t\bE_t\Big[\1_{\{ \tau^R  =\tau_2\neq \tau_1\leq T\}}B_{\tau_2}^{-1}(1-R_2)
[\1_{C_{\tau_2} < 0}((S^\Delta_{\tau_2}-R^h_{2}C_{\tau_2})^{-}-(S^{\Delta}_{\tau^R}-C_{\tau^R})^{-})]\\[0.05in]
& \hspace{-1cm} + \1_{\{ \tau^R  =\tau_2 = \tau_1\leq T\}}B_{\tau_2}^{-1}(1-R_2)
[\1_{C_{\tau_2} > 0}((S^\Delta_{\tau_2}-R^h_1C_{\tau_2})^{-}-(S^{\Delta}_{\tau^R}-C_{\tau^R})^{-})] \\[0.05in]
& \hspace{-0.5cm} +\1_{C_{\tau_2} < 0}((S^\Delta_{\tau_2}-R^h_{2}C_{\tau_2})^{-}-(S^{\Delta}_{\tau^R}-C_{\tau^R})^{-})\Big] \, ,
\end{align*}
for $t \in [0, \tau^R \wedge T]$.
\end{lemma}
\begin{proof}
Using \eqref{CVA} and \eqref{eq:CVARh} we obtain
\begin{align*}
\text{CVA}_t-\text{CVA}_t^{R,h} &= B_t\bE_t\Big[\1_{\{ \tau =\tau_1\leq T\}}B_{\tau}^{-1}(1-R_1)(S^\Delta_{\tau}-C_{\tau })^{+}\Big]  \\[0.05in]
& \quad - B_t\bE_t\Big[\1_{\{ \tau =\tau_2\leq T\}}B_{\tau}^{-1}(1-R_2)(S^\Delta_{\tau}-C_{\tau })^{-}\Big]  \\[0.05in]
& \quad - B_t\bE_t\Big[\1_{\{ \tau^R =\tau_1\leq T\}}B_{\tau^R}^{-1}(1-R_1)(S^\Delta_{\tau^R}-\widetilde{C}^1_{\tau^R})^{+}\Big]  \\[0.05in]
& \quad + B_t\bE_t\Big[\1_{\{ \tau^R  =\tau_2\leq T\}}B_{\tau^R}^{-1}(1-R_2)(S^\Delta_{\tau^R}-\widetilde{C}^2_{\tau^R})^{-}\Big] \, ,
\end{align*}
where
\begin{align*}
  \widetilde{C}^1_{\tau^R} &=
\1_{{\tau^R}=\tau_1 \neq \tau_{2}}(R^h_1C^+_{\tau^R}+C^-_{\tau^R})
 + \1_{{\tau^R}=\tau_{1}=\tau_2}(R^h_1C^+_{\tau^R}+R^h_{2}C^-_{\tau^R}) \, ,
\end{align*}
and
\begin{align*}
  \widetilde{C}^2_{\tau^R} &=
 \1_{{\tau^R}=\tau_{2} \neq \tau_1}(C^+_{\tau^R}+R^h_{2}C^-_{\tau^R})
 + \1_{{\tau^R}=\tau_{1}=\tau_2}(R^h_1C^+_{\tau^R}+R^h_{2}C^-_{\tau^R}).
\end{align*}
Rearranging the terms above yields
\begin{align}
\text{CVA}_t-\text{CVA}_t^{R,h} &= B_t\bE_t\Big[\1_{\{ \tau =\tau_1\leq T\}}B_{\tau_1}^{-1}(1-R_1)(S^\Delta_{\tau_!}-C_{\tau_1 })^{+}\Big] \label{eq:RVAh} \\[0.05in]
& \quad - B_t\bE_t\Big[\1_{\{ \tau =\tau_2\leq T\}}B_{\tau_2}^{-1}(1-R_2)(S^\Delta_{\tau_2}-C_{\tau_2})^{-}\Big] \nonumber \\[0.05in]
& \quad - B_t\bE_t\Big[\1_{\{ \tau^R =\tau_1\leq T\}}B_{\tau_1}^{-1}(1-R_1)(S^\Delta_{\tau_1}-\widetilde{C}^1_{\tau_1})^{+}\Big] \nonumber  \\[0.05in]
& \quad + B_t\bE_t\Big[\1_{\{ \tau^R  =\tau_2\leq T\}}B_{\tau_2}^{-1}(1-R_2)(S^\Delta_{\tau_2}-\widetilde{C}^2_{\tau_2})^{-}\Big] \nonumber \, .
\end{align}
Plugging in the terms $\widetilde{C}^1_{\tau_1}$ and $\widetilde{C}^2_{\tau_2}$ into \eqref{eq:RVAh}, we get
\begin{align*}
\text{CVA}_t-\text{CVA}_t^{R,h} &= B_t\bE_t\Big[\1_{\{ \tau =\tau_1\leq T\}}B_{\tau_1}^{-1}(1-R_1)(S^\Delta_{\tau_!}-C_{\tau_1 })^{+}\Big]  \\[0.05in]
& \hspace{-2.7cm} - B_t\bE_t\Big[\1_{\{ \tau =\tau_2\leq T\}}B_{\tau_2}^{-1}(1-R_2)(S^\Delta_{\tau_2}-C_{\tau_2})^{-}\Big]  \\[0.05in]
& \hspace{-2.7cm} - B_t\bE_t\Big[\1_{\{ \tau^R =\tau_1 \neq \tau_{2}\leq T\}}B_{\tau_1}^{-1}(1-R_1)
[\1_{C_{\tau_1} > 0}(S^\Delta_{\tau_1}-R^h_1C_{\tau_1})^+
+\1_{C_{\tau_1} < 0}(S^\Delta_{\tau_1}-C_{\tau_1})^{+})]\\[0.05in]
& \hspace{-1.7cm} +\1_{\{ \tau^R =\tau_1 = \tau_{2}\leq T\}}B_{\tau_1}^{-1}(1-R_1)
[\1_{C_{\tau_1} > 0}(S^\Delta_{\tau_1}-R^h_1C_{\tau_1})^+
+\1_{C_{\tau_1} < 0}(S^\Delta_{\tau_1}-R^h_{2}C_{\tau_1})^{+})]\Big]  \\[0.05in]
& \hspace{-2.7cm} + B_t\bE_t\Big[\1_{\{ \tau^R  =\tau_2\neq \tau_1\leq T\}}B_{\tau_2}^{-1}(1-R_2)
[\1_{C_{\tau_2} > 0}(S^\Delta_{\tau_2}-C_{\tau_2})^{-}
+\1_{C_{\tau_2} < 0}(S^\Delta_{\tau_2}-R^h_{2}C_{\tau_2})^{-}]\\[0.05in]
& \hspace{-1.7cm} + \1_{\{ \tau^R  =\tau_2 = \tau_1\leq T\}}B_{\tau_2}^{-1}(1-R_2)
[\1_{C_{\tau_2} > 0}(S^\Delta_{\tau_2}-R^h_1C_{\tau_2})^{-}]
+\1_{C_{\tau_2} < 0}(S^\Delta_{\tau_2}-R^h_{2}C_{\tau_2})^{-}\Big] \, .
\end{align*}
It follows from \eqref{RVA} that
\begin{align*}
\text{CVA}_t-\text{CVA}_t^{R,h} &= \text{RVA}_t +
 B_t \bE_t\Big[\1_{\{ \tau^R=\tau_1\leq T\}}B_{\tau^R}^{-1}(1-R_1)(S^{\Delta}_{\tau^R}-C_{\tau^R})^{+}\Big]  \\[0.05in]
&\hspace{-2.7cm} - B_t \bE_t\Big[\1_{\{ \tau^R=\tau_2\leq T\}}B_{\tau^R}^{-1}(1-R_2)(S^{\Delta}_{\tau^R}-C_{\tau^R})^{-}\Big]\\[0.05in]
& \hspace{-2.7cm} - B_t\bE_t\Big[\1_{\{ \tau^R =\tau_1 \neq \tau_{2}\leq T\}}B_{\tau_1}^{-1}(1-R_1)
[\1_{C_{\tau_1} > 0}(S^\Delta_{\tau_1}-R^h_1C_{\tau_1})^+
+\1_{C_{\tau_1} < 0}(S^\Delta_{\tau_1}-C_{\tau_1})^{+})]\\[0.05in]
& \hspace{-1.7cm}+\1_{\{ \tau^R =\tau_1 = \tau_{2}\leq T\}}B_{\tau_1}^{-1}(1-R_1)
[\1_{C_{\tau_1} > 0}(S^\Delta_{\tau_1}-R^h_1C_{\tau_1})^+
+\1_{C_{\tau_1} < 0}(S^\Delta_{\tau_1}-R^h_{2}C_{\tau_1})^{+})]\Big]  \\[0.05in]
& \hspace{-2.7cm} + B_t\bE_t\Big[\1_{\{ \tau^R  =\tau_2\neq \tau_1\leq T\}}B_{\tau_2}^{-1}(1-R_2)
[\1_{C_{\tau_2} > 0}(S^\Delta_{\tau_2}-C_{\tau_2})^{-}
+\1_{C_{\tau_2} < 0}(S^\Delta_{\tau_2}-R^h_{2}C_{\tau_2})^{-}]\\[0.05in]
& \hspace{-1.7cm} + \1_{\{ \tau^R  =\tau_2 = \tau_1\leq T\}}B_{\tau_2}^{-1}(1-R_2)
[\1_{C_{\tau_2} > 0}(S^\Delta_{\tau_2}-R^h_1C_{\tau_2})^{-}]
+\1_{C_{\tau_2} < 0}(S^\Delta_{\tau_2}-R^h_{2}C_{\tau_2})^{-}\Big] \, .
\end{align*}
Finally, we find
\begin{align*}
\text{CVA}_t-\text{CVA}_t^{R,h} &= \text{RVA}_t \\[0.05in]
& \hspace{-2cm}+ B_t\bE_t\Big[\1_{\{ \tau^R =\tau_1 \neq \tau_{2}\leq T\}}B_{\tau_1}^{-1}(1-R_1)
[\1_{C_{\tau_1} > 0}((S^{\Delta}_{\tau^R}-C_{\tau^R})^{+}-(S^\Delta_{\tau_1}-R^h_1C_{\tau_1})^+)]\\[0.05in]
& \hspace{-1cm} +\1_{\{ \tau^R =\tau_1 = \tau_{2}\leq T\}}B_{\tau_1}^{-1}(1-R_1)
[\1_{C_{\tau_1} > 0}((S^{\Delta}_{\tau^R}-C_{\tau^R})^{+}-(S^\Delta_{\tau_1}-R^h_1C_{\tau_1})^+) \\[0.05in]
& \hspace{-0.5cm} +\1_{C_{\tau_1} < 0}((S^{\Delta}_{\tau^R}-C_{\tau^R})^{+}-(S^\Delta_{\tau_1}-R^h_{2}C_{\tau_1})^{+})]\Big]  \\[0.05in]
& \hspace{-2cm} + B_t\bE_t\Big[\1_{\{ \tau^R  =\tau_2\neq \tau_1\leq T\}}B_{\tau_2}^{-1}(1-R_2)
[\1_{C_{\tau_2} < 0}((S^\Delta_{\tau_2}-R^h_{2}C_{\tau_2})^{-}-(S^{\Delta}_{\tau^R}-C_{\tau^R})^{-})]\\[0.05in]
& \hspace{-1cm} + \1_{\{ \tau^R  =\tau_2 = \tau_1\leq T\}}B_{\tau_2}^{-1}(1-R_2)
[\1_{C_{\tau_2} > 0}((S^\Delta_{\tau_2}-R^h_1C_{\tau_2})^{-}-(S^{\Delta}_{\tau^R}-C_{\tau^R})^{-})] \\[0.05in]
& \hspace{-0.5cm} +\1_{C_{\tau_2} < 0}((S^\Delta_{\tau_2}-R^h_{2}C_{\tau_2})^{-}-(S^{\Delta}_{\tau^R}-C_{\tau^R})^{-})\Big] \, ,
\end{align*}
for $t \in [0,\tau^R \wedge T]$.
\end{proof}

\begin{remark}
  Note that RVA$^{h}$ can either be negative or positive. If the difference is positive, then there is a decrease in the bilateral CVA, however if it is negative then there is an increase in the bilateral CVA.
\end{remark}

Let us define
\begin{align*}
  \text{URVA}^h_t:& =
   B_t\bE_t\Big[\1_{\{ \tau^R =\tau_1 \neq \tau_{2}\leq T\}}B_{\tau_1}^{-1}(1-R_1)
[\1_{C_{\tau_1} > 0}((S^{\Delta}_{\tau^R}-C_{\tau^R})^{+}-(S^\Delta_{\tau_1}-R^h_1C_{\tau_1})^+)]\\[0.05in]
& \hspace{-1cm} +\1_{\{ \tau^R =\tau_1 = \tau_{2}\leq T\}}B_{\tau_1}^{-1}(1-R_1)
[\1_{C_{\tau_1} > 0}((S^{\Delta}_{\tau^R}-C_{\tau^R})^{+}-(S^\Delta_{\tau_1}-R^h_1C_{\tau_1})^+) \\[0.05in]
& \hspace{-0.5cm} +\1_{C_{\tau_1} < 0}((S^{\Delta}_{\tau^R}-C_{\tau^R})^{+}-(S^\Delta_{\tau_1}-R^h_{2}C_{\tau_1})^{+})]\Big] \, ,  \\[0.05in]
  \text{DRVA}^h_t:& =
  B_t\bE_t\Big[\1_{\{ \tau^R  =\tau_2\neq \tau_1\leq T\}}B_{\tau_2}^{-1}(1-R_2)
[\1_{C_{\tau_2} < 0}((S^\Delta_{\tau_2}-R^h_{2}C_{\tau_2})^{-}-(S^{\Delta}_{\tau^R}-C_{\tau^R})^{-})]\\[0.05in]
& \hspace{-1cm} + \1_{\{ \tau^R  =\tau_2 = \tau_1\leq T\}}B_{\tau_2}^{-1}(1-R_2)
[\1_{C_{\tau_2} > 0}((S^\Delta_{\tau_2}-R^h_1C_{\tau_2})^{-}-(S^{\Delta}_{\tau^R}-C_{\tau^R})^{-})] \\[0.05in]
& \hspace{-0.5cm} +\1_{C_{\tau_2} < 0}((S^\Delta_{\tau_2}-R^h_{2}C_{\tau_2})^{-}-(S^{\Delta}_{\tau^R}-C_{\tau^R})^{-})\Big] \, .
\end{align*}
for $t \in [0,\tau^R \wedge T]$. Therefore, RVA$^h$ has the following decomposition,
\begin{align*}
\text{RVA}^h_t &= \text{RVA}_t + \text{URVA}^h_t + \text{DRVA}^h_t \, ,
\end{align*}
for $t \in [0,\tau^R \wedge T]$.

Here URVA$^h$ represents the expected loss if the counterparty defaults first which is preceded by a rating trigger.
Similarly, DRVA$^h$ is the expected loss in case the investor defaults first after a rating trigger.
Therefore, including rating triggers provision in an OTC contract provides protection from losses due to default events which happen after a credit downgrade. Accordingly, the value of the contract is adjusted for this protection, as shown in the following

\begin{corollary}
We have the following decomposition for the counterparty risk-free price process
\begin{align*}
S_t &=   S_t^{R,h} + \text{CVA}^{R,h}_t \\
&=S_t^{R,h} + \text{CVA}_t- \text{RVA}^h_t  \\
&=S_t^{R,h} + \text{UCVA}_t - \text{DVA}_t - \text{RVA}^h_t \\
& = S_t^{R,h} + \text{UCVA}_t - \text{DVA}_t - \text{RVA}_t - \text{URVA}^h_t - \text{DRVA}^h_t \, ,
\end{align*}
for $t \in [0,\tau^R \wedge T ]$.
\end{corollary}

\section{Markovian Approach for Rating-Based Pricing}\label{section:Copulae}

In this section, we employ Markov copulae for modeling the rating transitions in our framework.
Our approach is based on the studies of Bielecki et al. \cite{BCJRbasket2006,BieleckiEfficient2006,Bielecki2008a, BieleckiJVV2008, BJN2011}.

\subsection{Markov Copulae for the Multivariate Markov Chains}
Let us first consider two Markov chains ${X}^1$ and ${X}^2$ on $(\Omega,\cF,\bP)$ with the infinitesimal generators
$A^1:=[a^1_{ij}]$ and $A^2:=[a^2_{hk}]$, respectively.

In what follows, we work under the following assumption.
\begin{assumption}\label{asmp:copula}
\textup{The system of equations,
\begin{eqnarray}
\sum_{k \in \cK}a^X_{ih,jk} &  =& a^1_{ij} \, , \ \quad \forall i,j,h \in \cK,\; i \neq j,\;  \label{Eq:Mcop1}  \\[0.1in]
\sum_{j \in \cK}a^X_{ih,jk} & =& a^2_{hk} \, , \quad \forall i,h,k \in \cK,\; h \neq k,\; \label{Eq:Mcop2}
\end{eqnarray}
has a positive solution.}
\end{assumption}
The proof of the following proposition can be found in \cite{Bielecki2008a}.

\begin{proposition}
If Assumption~\ref{asmp:copula} is satisfied, then
$A^X=[a^X_{ih,jk}]_{i,h,j,k \in \cK}$ (where diagonal elements are defined appropriately) satisfies the conditions for a generator matrix of a bivariate time-homogeneous Markov chain, say $X=(Y^1,Y^2)$, whose components are Markov chains with the same laws as ${X}^1$ and ${X}^2$.
%
%
\end{proposition}

Hence, the resulting matrix $A^X=[a^X_{ih,jk}]_{i,h,j,k \in \cK}$ satisfies the conditions for a generator matrix of a bivariate time-homogeneous Markov chain, whose marginals are Markov chains with the same distributions as ${X}^1$ and ${X}^2$.
Therefore, the system \eqref{Eq:Mcop1}--\eqref{Eq:Mcop2} serves as a Markov copula between the Markovian margins ${Y}^1$, ${Y}^2$ and the bivariate Markov chain ${X}$.

Note that the system \eqref{Eq:Mcop1}--\eqref{Eq:Mcop2} can contain more unknowns than the number of equations, therefore being underdeteremined.
Therefore, as it is proposed by Bielecki et al. \cite{Bielecki2008a}, we impose additional constraints on the variables in the system \eqref{Eq:Mcop1}--\eqref{Eq:Mcop2}.
We postulate that
\begin{equation}
a^X_{ih,jk}=
\begin{cases}\label{Eq:alpha}
  0 \, ,&  \text{if} \; i \neq j, h \neq k, j \neq k  \\[0.05in]
  \alpha \min(a^1_{ij},a^2_{hk}) \, ,& \text{if} \; i \neq j, h \neq k, j=k
\end{cases}
\end{equation}
where $\alpha \in [0,1]$. 
Using the constraints \eqref{Eq:alpha} the system  \eqref{Eq:Mcop1}--\eqref{Eq:Mcop2} becomes fully decoupled, and we can obtain the generator of the joint process.

We interpret the constraint \eqref{Eq:alpha} as follows.
${Y}^1$ and ${Y}^2$ migrate according to their marginal laws.
Nevertheless, they can have the same values.
The intensity of migrating to the same rating category is measured by the parameter $\alpha$.
If $\alpha = 0$, then the components ${Y}^1$ and ${Y}^2$ of ${X}$ migrate independently.
However, if $\alpha = 1$, the tendency of ${Y}^1$ and ${Y}^2$ migrating to the same categories is at maximum.

\subsection{Markovian Changes of Measure}
Since rating transition matrices indicate the historical default probabilities, we need switch to the risk-neutral probabilites.
In practice, the change of measure is done such a way that the resulting risk-neutral probabilities are consistent with the default probabilities inferred from the quoted CDS spreads.
We need to apply changes of measure, while preserving Markovian structure of the model ${X}$.
Therefore, the process ${X}$, which is Markovian under the statistical measure, will remain Markovian under the risk-neutral measure as well.

Let $Y$ be a Markov process under $\bP$ with generator $A$ and domain $\cD(A)$ and define
\begin{equation*}
  M^f_t := \frac{f(Y_t)}{f(Y_0)}e^{-\int_0^t{\frac{Af(Y_s)}{f(Y_s)}ds}} \, .
\end{equation*}

The following definition is borrowed from \cite{palmowski02}.
\begin{definition}
  A strictly positive function $f \in \cD(A)$ is a good function if $M^f_t$ is a true (genuine) martingale with mean $1$ as $\bE_{\bP}(M^f_t)=1$.
\end{definition}
Let $f \in \cD(A)$ and $h$ be a good function and define
\begin{equation*}
  A^hf :=  h^{-1}A(fh) - fA(h) \, .
\end{equation*}

The proof of the following Theorem can be found in \cite{palmowski02}.
\begin{theorem}
  Let $\bQ^h$ be the probability measure associated to the density process $M^h_t$. Then $Y$ is a Markov process under $\bQ^h$ with extended generator $(A^h,\cD(A))$.
\end{theorem}

If $Y$ is a finite state Markov chain, then we have the following result.
\begin{corollary}
  Let $Y$ be a finite state Markov chain on $\cK$ with cardinality $\cK$ and generator $A = a_{ij}$ and let $h = (h_1,\dots,h_K)$ be a positive vector. Then $Y$ is a Markov process under $\bQ^h$ with generator $A^h = [a_{ij}h_jh^{-1}_i]$.
\end{corollary}

Using the above corollary, we can change the measure from the statistical measure $\bP$ to a risk-neutral measure $\bQ$ using a vector $h = (h_{11},h_{22}\dots,h_{KK}) \in \bR^{K}$, so that the process ${X}$ will be a time-homogeneous Markov chain under $\bQ$. In this case, the infinitesimal generator under $\bQ$ is found as
\begin{equation*}
  {A}^X = [{a}_{ih,jk}] \, ,
\end{equation*}
where
\begin{equation*}
{a}_{ih,jk}:=
\begin{cases}
a_{ih,jk} \frac{h_{jk}}{h_{ih}} & \text{if $ih \neq jk$} \, ,  \\[0.05in]
-\sum_{ih \neq jk}a_{ih,jk} \frac{h_{jk}}{h_{ih}} & \text{if $ih = jk$} \, .
\end{cases}
\end{equation*}

In Bielecki et al \cite{Bielecki2008a}, it is suggested that the vector $h_{ij}$ can be chosen as
\begin{equation*}
  h_{ij} = e^{\alpha_1i +\alpha_2j} \, ,  \qquad i,j \in \cK \, ,
\end{equation*}
where the parameters $\alpha_1$ and $\alpha_2$ can be estimated through calibration.

\section{Applications}\label{section:Apps}
In this section, we illustrate our results in the context of a CDS and an IRS contract.
We postulate that our CDS and IRS contracts are subject to rating triggers, so that they are terminated in case a trigger event occurs.
We compute the adjustments we discussed previously; namely, CVA, DVA, URVA and DRVA of the contracts for different rating trigger levels.
Moreover, we compare CVA$^R$ and CVA values and find the impact of adding rating triggers on the adjustments.

For the sake of simplicity, we carry out our analysis with $K = 4$ rating categories: A, B, C and D.
The level A represents the highest rating level, whereas D corresponds to the default state.
We assume that the counterparty initially has rating $A$. In what follows, we suppose that the 1-year rating transition matrix is given in Table \ref{table:P1}.
\begin{table}[!h]
\caption{Counterparty's rating transition matrix: $P^1$}
\smallskip
\centering
\begin{tabular}{cccccc}
\hline \hline
  & A & B & C & D \\[0.05in] \hline
  A & 0.9 & 0.08 & 0.017 & 0.003  \\[0.05in]
  B & 0.05 & 0.85 & 0.09 & 0.01  \\[0.05in]
  C & 0.01 & 0.09 & 0.8 & 0.1 \\[0.05in]
  D &  0 & 0 & 0 & 1  \\[0.05in]
  \hline
\end{tabular}
\label{table:P1}
\end{table}

Moreover, we assume that the current rating of the investor is $A$. Investor's 1-year rating transition matrix is assumed to be given as in Table \ref{table:P2}.
\begin{table}[!h]
\caption{Investor's rating transition matrix: $P^2$}
\smallskip
\centering
\begin{tabular}{cccccc}
\hline \hline
  & A & B & C & D \\[0.05in] \hline
 A & 0.8 & 0.1 & 0.05 & 0.05  \\[0.05in]
  B  & 0.04 & 0.9 & 0.03 & 0.03  \\[0.05in]
   C & 0.015 & 0.1 & 0.7 & 0.185  \\[0.05in]
   D & 0 & 0 & 0 & 1  \\[0.05in]
  \hline
\end{tabular}
\label{table:P2}
\end{table}

We assume that the rating transition matrices given above are already risk-neutral, therefore we set $\alpha_1 = \alpha_2 = 0$.
We also assume deterministic recovery rates; $R_1 = R_2 =0.4$ and $R^h_1 = R^h_2 = 1$.

\subsection{CVA of an IRS with Rating Triggers}
In this section, we compute the CVA, DVA, and RVA of a fixed-for-float payer $10$-year IRS contract with \$1 notional, in presence of rating triggers as break clauses.
We assume that the payments are done every quarter, and the fixed leg pays the swap rate, while the floating leg pays the LIBOR rate.
We also assume that the swap is initiated at $T_0:=0$ and we denote by $T_1 < T_2 < \dots < T_n$, the collection of payment dates and $S$ by the fixed rate.

As we noted above, the rating transition matrices of the counterparty and the investor is given by $P^1$ and $P^2$, respectively.

The cumulative dividend process of the IRS contract at time $T_i$ is given by
\begin{align*}
  D_{T_i} &= \sum_{k=1}^i (L(T_k)-S)\delta_k \, ,
\end{align*}
where $L(T_i)$ is time-$T_i$ LIBOR rate and $\delta_k = T_k - T_{k-1}$ for $k = 1,2,\dots,n$. We also suppose that the instantaneous interest rate $r$ follows
\begin{align*}
  dr_t = (\theta - \alpha \, r_t)dt + \sigma dW_t
\end{align*}
where we set $r_0=0.05$, $\theta = 0.1$, $\alpha = 0.05$ and $\sigma = 0.01$. We find the corresponding swap rate as $S = 0.0496$.

We carry out our analysis for uncollateralized, linearly collateralized and exponentially collateralized cases for $\alpha = 0$ and $\alpha = 1$.
Our results are displayed in Tables \ref{table:IRS1}--\ref{table:IRS6}.
We observe that the initial URVA values decrease with the decreasing counterparty trigger levels, which we denote by $K_1$.
Similarly, the initial DRVA values also decrease when we decrease the investor's trigger level, which is denoted by $K_2$.
However, the RVA values, which indicate the total bilateral adjustment due to the additional rating triggers, do not follow a certain pattern.
For example, in Table \ref{table:IRS1}, although we decrease the trigger levels from $K_1 = B$, $K_2 = B$ to $K_1 = C$, $K_2 = C$, the corresponding RVA values do not necessarily decrease, as opposed to URVA and DRVA values.
We also observe from in Tables \ref{table:IRS1}--\ref{table:IRS6} that adding bilateral rating triggers can actually decrease the initial bilateral CVA values (in absolute values), compared to the case with no rating triggers, which is $K_1 = D$ or $K_2 = D$.
For instance, in Table \ref{table:IRS1}, the absolute value of CVA$^R$ with no rating triggers is almost three times greater than the absolute value of CVA$^R$ with $K_1 = B$ and $K_2 = B$.
In other words, in this case adding rating triggers decreases the absolute value of the bilateral CVA by nearly 60\%.
Detailed information about the change in the CVA$^R$ values due to rating triggers can be seen in Figures \ref{CVArating1}--\ref{CVArating3a}
Also, it can be seen from the Tables \ref{table:IRS2},\ref{table:IRS4} and \ref{table:IRS6}, where $\alpha = 1$, that the URVA and DRVA values are slightly higher compared to the values in Tables \ref{table:IRS1},\ref{table:IRS3} and \ref{table:IRS5}, where $\alpha = 0$.

\begin{table}[tH]
\caption{CVA and RVA (1\$ $\times 10^{-3}$) components of an IRS, $\alpha=0$, No collateralization}
\smallskip
\centering
\renewcommand{\arraystretch}{1.3}
\begin{tabular}{cccccccc}
\hline \hline
    $K_1$ & $K_2$ & URVA  & DRVA & RVA & UCVA$^R$ & DVA$^R$ & CVA$^R$ \\
    \hline
  B & B & 1.21214  & 7.46661  & -6.25446  & 0.20602 & 3.66672  & -3.46070  \\[0.03in]
  B & C & 1.13272  & 7.60683  & -6.47411  & 0.24496 & 3.94986  & -3.70490  \\[0.03in]
  C & B & 0.90323 & 7.47735  & -6.57412  & 0.51707 & 4.04024  & -3.52316  \\[0.03in]
  C & C & 0.91109  & 7.61633  & -6.70523  & 0.47848 & 3.90228  & -3.42379  \\[0.03in]
  B & D & 1.20710  & 0 & 1.20710  & 0.23716 & 11.1861 & -10.9490 \\[0.03in]
  D & B & 0 & 7.77754  & -7.77754  & 1.50820  & 3.48440  & -1.97620  \\[0.03in]
  C & D & 0.83628 & 0 & 0.83628 & 0.58417 & 11.4550 & -10.8708 \\[0.03in]
  D & C & 0 & 7.76986  & -7.76986  & 1.38051  & 3.87084  & -2.49033 \\[0.03in]
  D & D & 0 & 0 & 0 & 1.40033  & 11.4083 & -10.0080\\[0.03in]
  \hline
\end{tabular}
\label{table:IRS1}
\end{table}

\begin{table}[H]
\caption{Mitigation in the CVA of an IRS, $\alpha=0$, No collateralization}
\smallskip
\centering
\renewcommand{\arraystretch}{1.3}
\begin{tabular}{cccccccc}
\hline \hline
   (B,B)  & (B,C) & (C,B) &  (C,C) & (B,D) & (D,B) & (C,D) & (D,C)\\
    \hline
   65.42\% & 62.98 \%  & 64.80\% & 65.79\%  & -9.40\% &  80.25\% & -8.61\% & 75.12\% \\[0.03in]
  \hline
\end{tabular}
\label{table:IRS1g}
\end{table}

 \begin{figure}[H]
\centering
\begin{tabular}{c}
\epsfig{file=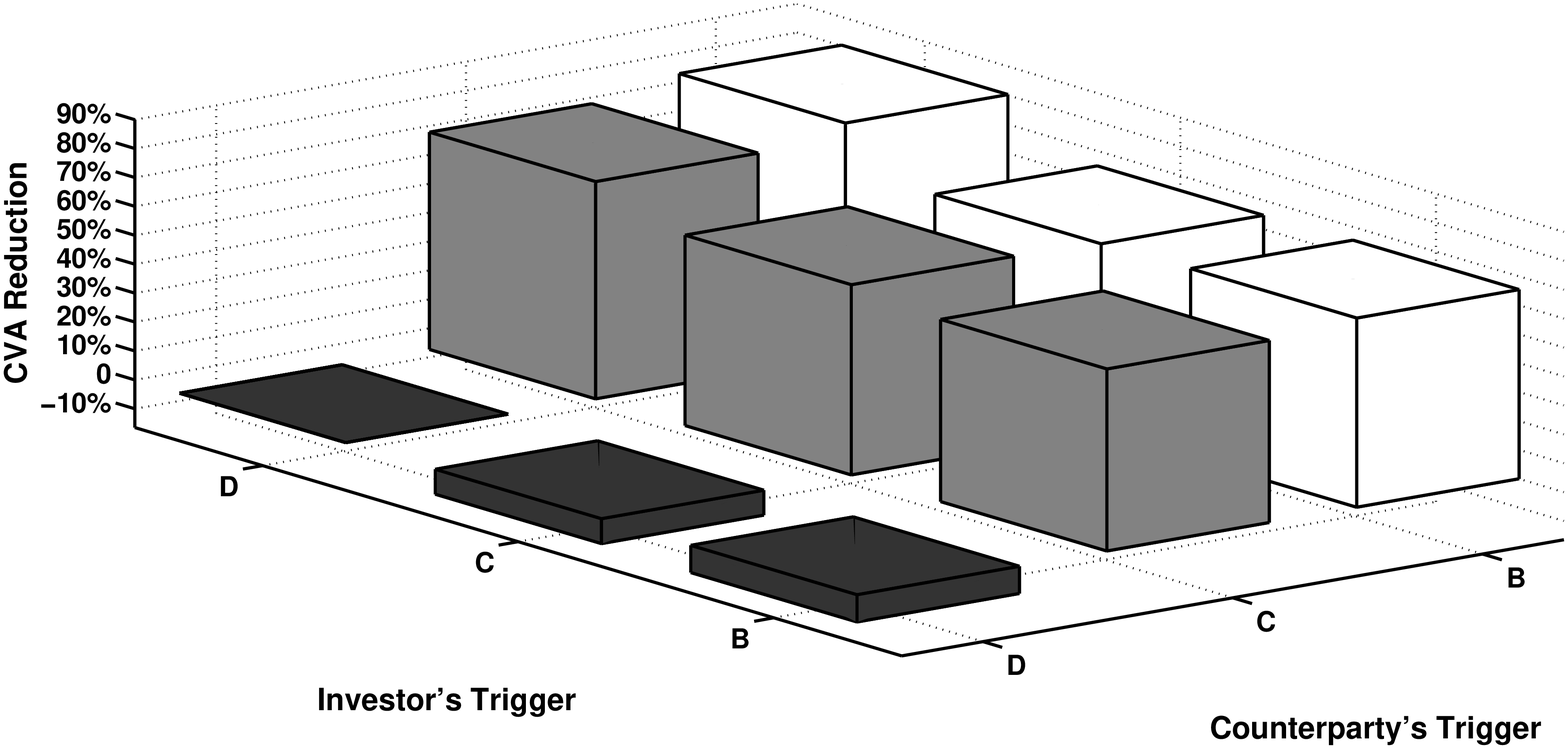,width=0.9\linewidth,clip=}
\end{tabular}
\caption{Mitigation in the CVA of an IRS (in \%), $\alpha=0$, No collateralization}
\label{CVArating1}
\end{figure}

\begin{table}[tH]
\caption{CVA and RVA (1\$ $\times 10^{-3}$) components of an IRS, $\alpha=1$, No collateralization}
\smallskip
\centering
\renewcommand{\arraystretch}{1.3}
\begin{tabular}{cccccccc}
\hline \hline
$K_1$ & $K_2$ & URVA  & DRVA & RVA & UCVA$^R$ & DVA$^R$ & CVA$^R$ \\
    \hline
  B & B & 1.56101  & 8.90492  & -7.34393  & 0.24788  & 3.44390  & -3.19610 \\[0.03in]
  B & C & 1.39405  & 8.25898  & -6.86493  & 0.14690  & 3.46560  & -3.31870 \\[0.03in]
  C & B & 1.01004  & 8.91694  & -7.90690  & 0.17087  & 3.39857  & -3.22770 \\[0.03in]
  C & C & 1.21133  & 8.26267  & -7.05133  & 0.49359  & 3.60699  & -3.11339 \\[0.03in]
  B & D & 1.48947  & 0 & 1.48947  & 0.03862 & 12.3056  & -12.2669\\[0.03in]
  D & B & 0 & 8.97150  & -8.97150  & 1.83154  & 3.28835  & -1.45680 \\[0.03in]
  C & D & 0.99948 & 0 & 0.99948 & 0.06116 & 12.4500 & -12.3889 \\[0.03in]
  D & C & 0 & 8.45006  & -8.45006  & 1.63383  & 3.15669  & -1.52285 \\[0.03in]
  D & D & 0 & 0 & 0 & 1.68138  & 12.3790 & -10.6977\\[0.03in]
  \hline
\end{tabular}
\label{table:IRS2}
\end{table}

\begin{table}[H]
\caption{Mitigation in the CVA of an IRS (in \%), $\alpha=1$, No collateralization}
\smallskip
\centering
\renewcommand{\arraystretch}{1.3}
\begin{tabular}{cccccccc}
\hline \hline
   (B,B)  & (B,C)   & (C,B)  &  (C,C)  & (B,D)   & (D,B)   & (C,D)   & (D,C)   \\
    \hline
   70.12\% & 68.98\%  & 69.83\% & 70.90\%  & -14.67\% &  86.38\% & -15.81\% & 85.76\%  \\[0.03in]
  \hline
\end{tabular}
\label{table:IRS2g}
\end{table}

 \begin{figure}[H]
\centering
\begin{tabular}{c}
\epsfig{file=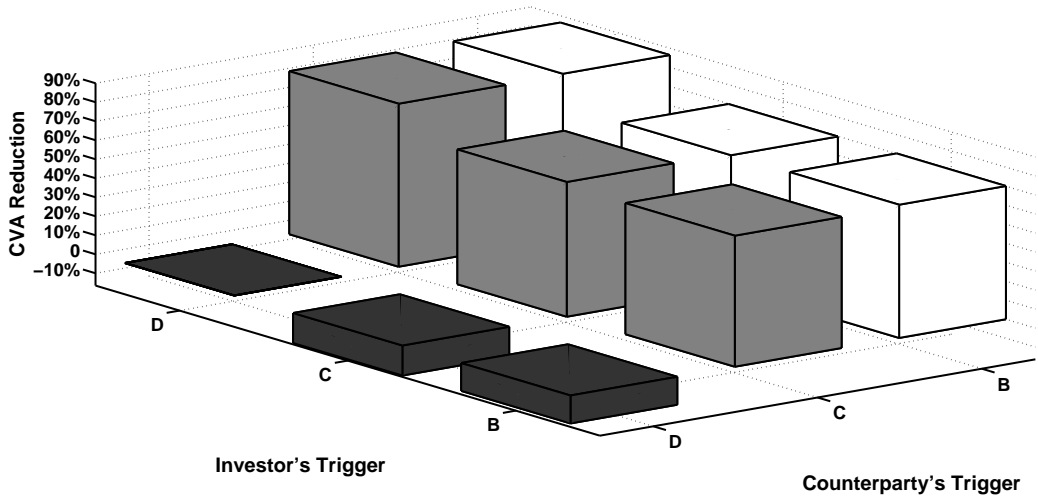,width=0.90\linewidth,clip=}
\end{tabular}
\caption{Mitigation in the CVA of an IRS (in \%), $\alpha=1$, No collateralization}
\label{CVArating1a}
\end{figure}
Moreover, we see that the UCVA$^R$ values start increasing as we lower the counterparty trigger levels. We also observe that the DVA$^R$ values also increase with the decreasing trigger levels for the investor. However, the CVA$^R$ values, that is the bilateral CVA, do not change significantly unless we set $K_1 = D$ or $K_2 = D$, which essentially means elimination of rating triggers.

\begin{table}[tH]
\caption{CVA and RVA (1\$ $\times 10^{-3}$) components of an IRS, $\alpha=0$, Linear collateral rate: $\rho^i_l$}
\smallskip
\centering
\renewcommand{\arraystretch}{1.3}
\begin{tabular}{cccccccc}
\hline \hline
$K_1$ & $K_2$ & URVA  & DRVA & RVA & UCVA$^R$ & DVA$^R$ & CVA$^R$ \\
    \hline
  B & B & 0.65278  & 4.25839  & -3.60560  & 0.16258  & 1.79760  & -1.63502  \\[0.03in]
  B & C & 0.60496  & 3.17065  & -2.56569  & 0.18720  & 3.02070  & -2.83349 \\[0.03in]
  C & B & 0.42969  & 4.23383  & -3.80413  & 0.42137  & 1.88225  & -1.46087 \\[0.03in]
  C & C & 0.43839  & 3.20526  & -2.76687  & 0.38831  & 3.01006  & -2.62174 \\[0.03in]
  B & D & 0.67613  & 0 & 0.67613  & 0.18604  & 6.04750  & -5.86146 \\[0.03in]
  D & B & 0 & 4.37787  & -4.37787  & 0.91087  & 1.68992  & -0.77911  \\[0.03in]
  C & D & 0.43038  & 0 & 0.43038  & 0.46165  & 6.08349  & -5.62184 \\[0.03in]
  D & C & 0 & 3.26057  & -2.43012  & 0.83045  & 2.97872  & -2.14826  \\[0.03in]
  D & D & 0 & 0 & 0 & 0.83666 & 6.11896  & -5.28229  \\[0.03in]
  \hline
\end{tabular}
\label{table:IRS3}
\end{table}

\begin{table}[H]
\caption{Mitigation in the CVA of an IRS (in \%), $\alpha=0$, Linear collateral rate: $\rho^i_l$}
\smallskip
\centering
\renewcommand{\arraystretch}{1.3}
\begin{tabular}{cccccccc}
\hline \hline
     (B,B)  & (B,C) & (C,B) &  (C,C) & (B,D) & (D,B) & (C,D) & (D,C)  \\
    \hline
   69.05\%  & 46.36\%   & 72.34\% &  50.37\%  & -10.97\% &  85.25\% & -6.43\% & 59.33\%  \\[0.03in]
  \hline
\end{tabular}
\label{table:IRS3g}
\end{table}

 \begin{figure}[H]
\centering
\begin{tabular}{c}
\epsfig{file=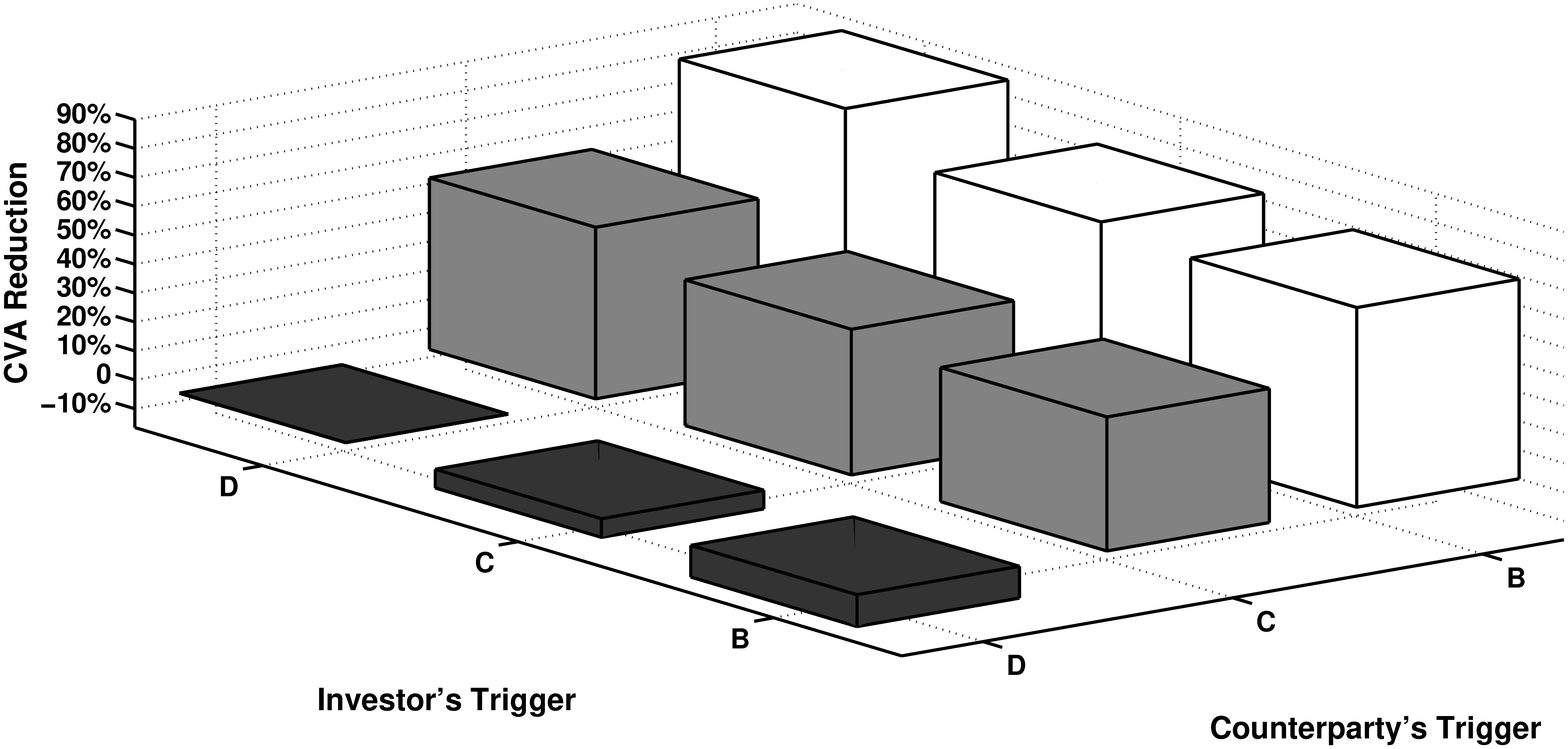,width=0.90\linewidth,clip=}
\end{tabular}
\caption{Mitigation in the CVA of an IRS (in \%), $\alpha=0$, Linear collateral rate: $\rho^i_l$}
\label{CVArating2}
\end{figure}

We note that DRVA values are equal to zero whenever $K_2 = D$. This is because by setting the investors trigger to level D, we simply do not to have any ratings adjustments for the investor. Similarly, we see that the URVA values are equal to zero where $K_1 = D$. Naturally, the case $K_1 = D$ and $K_2 = D$ corresponds to the CVA computation without any rating triggers.

\begin{table}[tH]
\caption{CVA and RVA (1\$ $\times 10^{-3}$) components of an IRS, $\alpha=1$, Linear collateral rate: $\rho^i_l$}
\smallskip
\centering
\renewcommand{\arraystretch}{1.3}
\begin{tabular}{cccccccc}
\hline \hline
$K_1$ & $K_2$ & URVA  & DRVA & RVA & UCVA$^R$ & DVA$^R$ & CVA$^R$ \\
    \hline
  B & B & 0.90873  & 4.95255  & -4.04382  & 0.21863  & 1.63945  & -1.42081 \\[0.03in]
  B & C & 0.77176  & 3.57184  & -2.80007  & 0.14690  & 2.65780  & -2.51090 \\[0.03in]
  C & B & 0.50178  & 4.97124  & -4.46946  & 0.17087  & 1.61150  & -1.44063 \\[0.03in]
  C & C & 0.61942  & 3.58533  & -2.96590  & 0.41343  & 2.74593  & -2.33249 \\[0.03in]
  B & D & 0.85003  & 0 & -0.85003  & 0.03862  & 6.60363  & -6.56500 \\[0.03in]
  D & B & 0 & 5.02376  & -5.02376  & 1.12596  & 1.52697  & -0.40100 \\[0.03in]
  C & D & 0.48810  & 0 & 0.48810  & 0.06116  & 6.68342  & -6.62226 \\[0.03in]
  D & C & 0 & 3.65312  & -3.65312  & 0.97908  & 2.41592  & -1.43683 \\[0.03in]
  D & D & 0 & 0 & 0 &  1.06488  & 6.75473  & -5.68985 \\[0.03in]
  \hline
\end{tabular}
\label{table:IRS4}
\end{table}

\begin{table}[H]
\caption{Mitigation in the CVA of an IRS (in \%), $\alpha=1$, Linear collateral rate: $\rho^i_l$}
\smallskip
\centering
\renewcommand{\arraystretch}{1.3}
\begin{tabular}{cccccccc}
\hline \hline
   (B,B)  & (B,C)   & (C,B)  &  (C,C)  & (B,D)   & (D,B)   & (C,D)   & (D,C)  \\
    \hline
   75.03\% & 55.84\%  & 74.68\% & 59.01\%  & -15.38\% &  92.95\% & -16.39\% & 74.75\% \\[0.03in]
  \hline
\end{tabular}
\label{table:IRS4g}
\end{table}

 \begin{figure}[H]
\centering
\begin{tabular}{c}
\epsfig{file=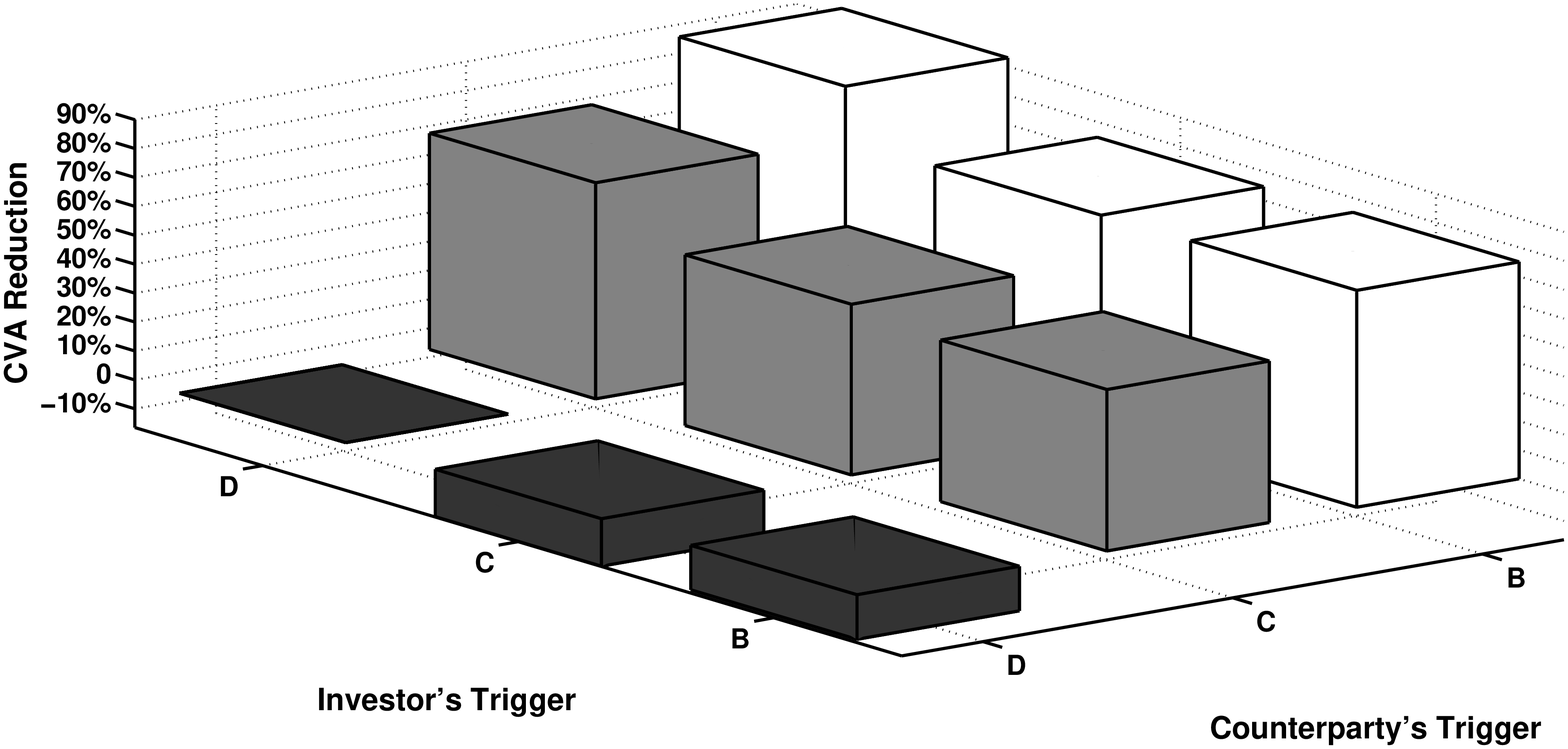,width=0.90\linewidth,clip=}
\end{tabular}
\caption{Mitigation in the CVA of an IRS (in \%), $\alpha=1$, Linear collateral rate: $\rho^i_l$}
\label{CVArating2a}
\end{figure}

\begin{table}[tH]
\caption{CVA and RVA (1\$ $\times 10^{-3}$) components of an IRS, $\alpha=0$, Exponential collateral rate: $\rho^i_e$}
\smallskip
\centering
\renewcommand{\arraystretch}{1.3}
\begin{tabular}{cccccccc}
\hline \hline
$K_1$ & $K_2$ & URVA  & DRVA & RVA & UCVA$^R$ & DVA$^R$ & CVA$^R$ \\
    \hline
  B & B & 0.37784  & 2.47245  & -2.09460  & 0.14968  & 1.24248  & -1.09280 \\[0.03in]
  B & C & 0.33067  & 1.54366  & -1.21298  & 0.17005  & 2.18783  & -2.01777 \\[0.03in]
  C & B & 0.22534  & 2.40717  & -2.18183  & 0.33559  & 1.24133  & -0.90573 \\[0.03in]
  C & C & 0.23310  & 1.57223  & -1.33913  & 0.30749  & 2.21030  & -1.90281 \\[0.03in]
  B & D & 0.38613  &0  & 0.38613  & 0.17085  & 3.69667  & -3.52582 \\[0.03in]
  D & B & 0 & 2.53262  & -2.53262  & 0.58359  & 1.15705  & -0.57345 \\[0.03in]
  C & D & 0.23556  & 0 & 0.23556  & 0.35182  &  3.65194  & -3.30011 \\[0.03in]
  D & C & 0 & 1.60587  & -1.60587  & 0.54372  & 2.17905  & -1.63533 \\[0.03in]
  D & D & 0 & 0 & 0 & 0.54108  & 3.70621  & -3.16512 \\[0.03in]
  \hline
\end{tabular}
\label{table:IRS5}
\end{table}

\begin{table}[H]
\caption{Mitigation in the CVA of an IRS (in \%), $\alpha=0$, Exponential collateral rate: $\rho^i_e$}
\smallskip
\centering
\renewcommand{\arraystretch}{1.3}
\begin{tabular}{cccccccc}
\hline \hline
     (B,B)   & (B,C)    & (C,B)  &  (C,C)     & (B,D)  & (D,B)    & (C,D)  & (D,C)  \\
    \hline
   65.47\%  & 36.25\%    & 71.38\%  &  39.88\%  & -11.40\% &  81.88\%  &   -4.27\% & 48.33\%  \\[0.03in]
  \hline
\end{tabular}
\label{table:IRS5g}
\end{table}

 \begin{figure}[H]
\centering
\begin{tabular}{c}
\epsfig{file=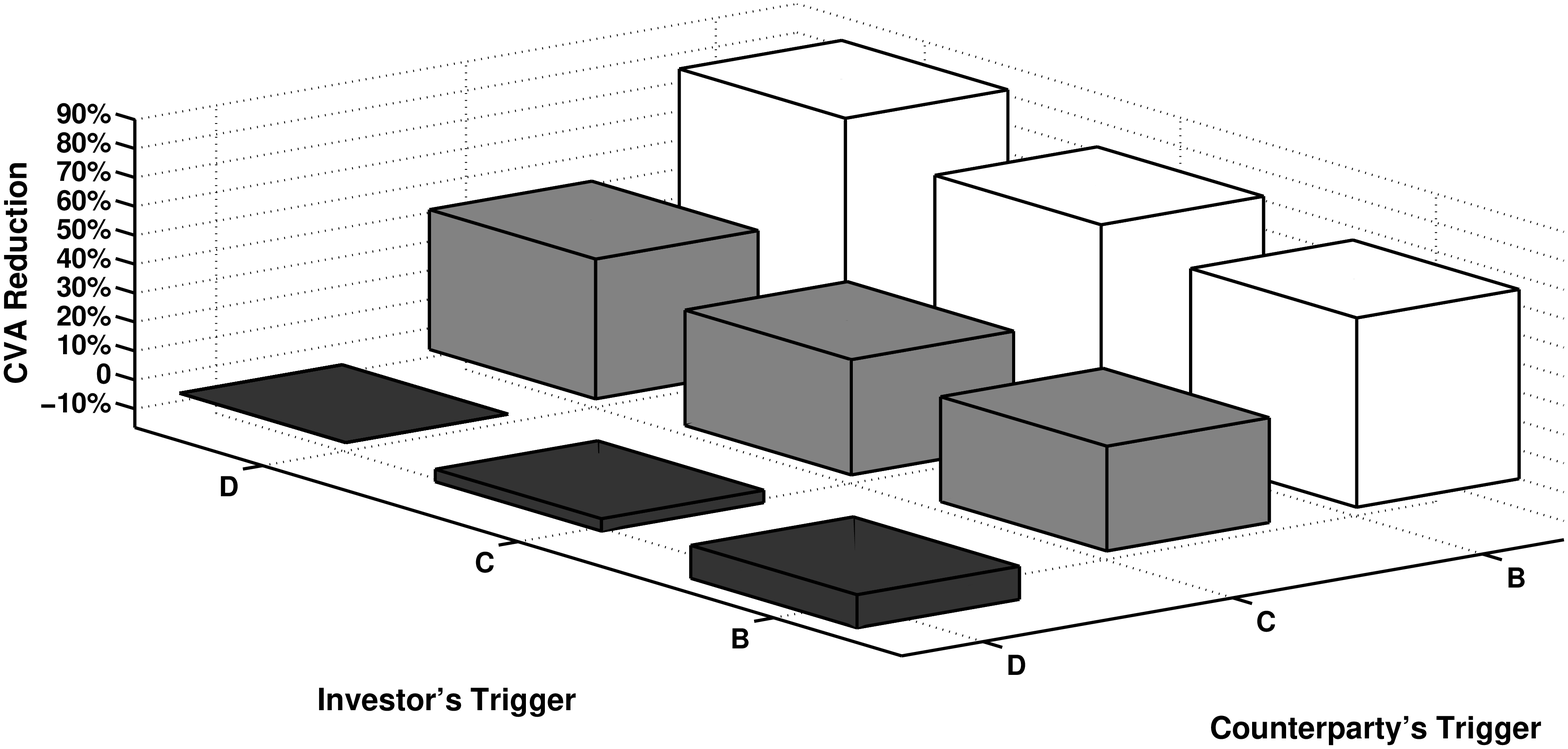,width=0.90\linewidth,clip=}
\end{tabular}
\caption{Mitigation in the CVA of an IRS (in \%), $\alpha=0$, Exponential collateral rate: $\rho^i_e$}
\label{CVArating3}
\end{figure}

\begin{table}[tH]
\caption{CVA and RVA (1\$ $\times 10^{-3}$) components of an IRS, $\alpha=1$, Exponential collateral rate: $\rho^i_e$}
\smallskip
\centering
\renewcommand{\arraystretch}{1.3}
\begin{tabular}[c]{cccccccc}
\hline \hline
$K_1$ & $K_2$ & URVA  & DRVA & RVA & UCVA$^R$ & DVA$^R$ & CVA$^R$ \\
    \hline
  B & B & 0.54239  & 2.8266  & -2.2842  & 0.20994  & 1.1035  & -0.89356 \\[0.03in]
  B & C & 0.43664  & 1.7910  & -1.3544  & 0.14690  & 1.9337  & -1.7868 \\[0.03in]
  C & B & 0.26953  & 2.8503  & -2.5807  & 0.17087  & 1.0807  & -0.90987 \\[0.03in]
  C & C & 0.34294  & 1.8238  & -1.4809  & 0.34158  & 1.9741  & -1.6325 \\[0.03in]
  B & D & 0.50140  & 0 & 0.50140  & 0.038628  & 3.9588  & -3.9202 \\[0.03in]
  D & B & 0 & 2.8784  & -2.8784  & 0.76099  & 1.0038  & -0.24285 \\[0.03in]
  C & D & 0.26414  & 0 & 0.26414  & 6.1162  & 3.9915  & -3.9304 \\[0.03in]
  D & C & 0 & 1.8357  & -1.8357  & 0.61978  & 1.7519  & -1.1321 \\[0.03in]
  D & D & 0 & 0 & 0 & 0.73056  & 4.0954  & -3.3648 \\[0.03in]
  \hline
\end{tabular}
\label{table:IRS6}
\end{table}

\begin{table}[H]
\caption{Mitigation in the CVA of an IRS (in \%), $\alpha=1$, Exponential collateral rate: $\rho^i_e$}
\smallskip
\centering
\renewcommand{\arraystretch}{1.3}
\begin{tabular}{cccccccc}
\hline \hline
   (B,B)  & (B,C)   & (C,B)  &  (C,C)  & (B,D)   & (D,B)   & (C,D)   & (D,C)  \\
    \hline
   73.44\% & 46.90\% & 72.96\% & 51.48\% & -16.50\% &  92.78\% & -16.81\% & 66.35\%  \\[0.03in]
  \hline
\end{tabular}
\label{table:IRS6g}
\end{table}

 \begin{figure}[H]
\centering
\begin{tabular}{c}
\epsfig{file=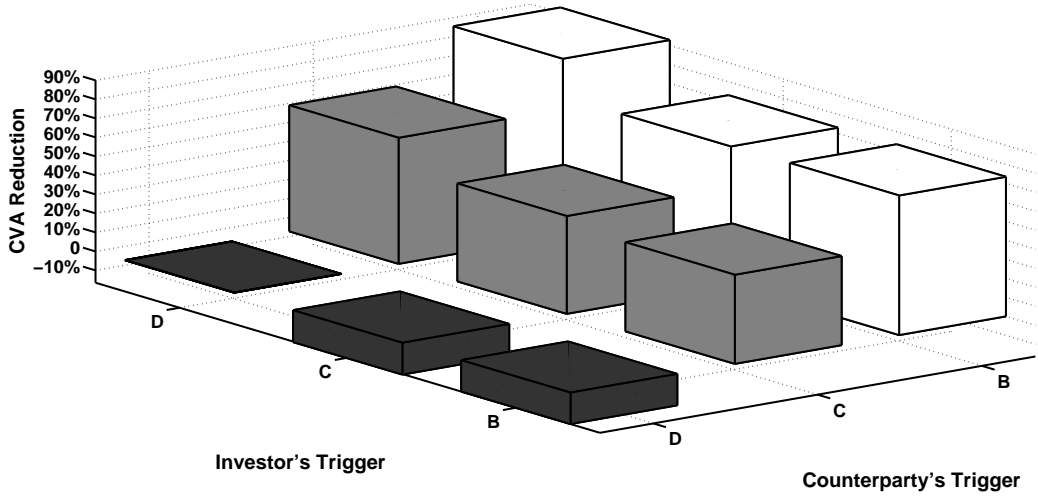,width=0.90\linewidth,clip=}
\end{tabular}
\caption{Mitigation in the CVA of an IRS (in \%), $\alpha=1$, Exponential collateral rate: $\rho^i_e$}
\label{CVArating3a}
\end{figure}

\subsection{CVA of a CDS with Rating Triggers}
In this section, we compute the CVA, DVA, and RVA of a CDS contract in presence of rating triggers as break clauses.
Recall that $D$ represents the counterparty risk-free cumulative dividend process of a contract.
We assume that the reference entity is free of any trigger events.
We denote by $\tau_3$ the default time of the reference entity and $R_3$ the recovery rate of the reference entity.
We assume that the CDS contract has spread $\kappa$, expires at $T$ and has nominal value of $1$.
Consequently, the cumulative dividend process of the CDS contract is given by
\begin{align*}
  D_t &= (1-R_3)\1_{\{\tau_3 \leq t \}} -  \kappa(t \wedge T \wedge \tau) \, ,
\end{align*}
for all $t \in [0,T]$.
We also assume that the underlying entity's 1-year rating transition matrix is given as in Table \ref{table:P3}.

\begin{table}[!h]
\caption{Underlying entity's rating transition matrix: $P^3$}
\smallskip
\centering
\begin{tabular}{cccccc}
\hline \hline
  & A & B & C & D \\[0.05in] \hline
 A & 0.95 & 0.03 & 0.019 & 0.001  \\[0.05in]
  B  & 0.04 & 0.85 & 0.107 & 0.003  \\[0.05in]
   C & 0.01 & 0.19 & 0.791 & 0.009  \\[0.05in]
   D & 0 & 0 & 0 & 1  \\[0.05in]
  \hline
\end{tabular}
\label{table:P3}
\end{table}

Similar to the IRS example, we carry out our analysis for uncollateralized, linearly collateralized and exponentially collateralized CDS contracts where $\alpha = 0$ and $\alpha = 1$.
We display our results in Tables \ref{table:CDS1}--\ref{table:CDS6}.

The initial URVA values increase with the increasing counterparty trigger levels, and the initial DRVA values increase with the increasing investor trigger levels.
However, the absolute values of the RVA numbers can increase or decrease with the changing trigger levels.
For example, in Table \ref{table:CDS1}, although we decrease the trigger levels from $K_1 = B$, $K_2 = B$ to $K_1 = C$, $K_2 = C$, the corresponding RVA values (in absolute terms) do not necessarily decrease, compared to the URVA and DRVA values.

It can also be observed from in Tables \ref{table:CDS1}--\ref{table:CDS6} that bilateral rating triggers can actually decrease the initial bilateral CVA values (in absolute values).
For instance, the absolute value of CVA$^R$ in Table \ref{table:CDS1} with no rating triggers is almost three times greater than the absolute value of CVA$^R$ with $K_1 = B$ and $K_2 = B$.
In addition, the UCVA$^R$ values with no rating triggers are also almost three times greater than the UCVA$^R$ values with $K_1 = B$ and $K_2 = B$. Similarly, the DVA$^R$ values with no rating triggers are almost four times greater than the UCVA$^R$ values with $K_1 = B$ and $K_2 = B$.
In other words, in in Table \ref{table:CDS1}, adding rating triggers decreases the UCVA$^R$ value by nearly 60\%, DVA$^R$ value by nearly 75\%, and the absolute value of the bilateral CVA$^R$ by nearly 60\%. Figures \ref{CVArating4}--\ref{CVArating6a} present the changes in the CVA values for each set of rating triggers.

Also, it can be seen from the Tables \ref{table:CDS2},\ref{table:CDS4} and \ref{table:CDS6}, where $\alpha = 1$, that the URVA and DRVA values are slightly higher compared to the values in Tables \ref{table:CDS1},\ref{table:CDS3} and \ref{table:CDS5}, where $\alpha = 0$.

Moreover, it can be seen from Tables \ref{table:CDS1}--\ref{table:CDS6} the UCVA$^R$ values start increasing as we lower the counterparty trigger levels. Similarly, the DVA$^R$ values also increase with the decreasing trigger levels for the investor. However, the CVA$^R$ values do not change significantly unless we set $K_1 = D$ or $K_2 = D$, or eliminate the rating triggers.

\begin{table}[tH]
\caption{CVA and RVA (1\$ $\times 10^{-3}$) components of a CDS, $\alpha=0$, No collateralization}
\smallskip
\centering
\renewcommand{\arraystretch}{1.3}
\begin{tabular}[c]{cccccccc}
\hline \hline
$K_1$ & $K_2$ & URVA  & DRVA & RVA & UCVA$^R$ & DVA$^R$ & CVA$^R$ \\
    \hline
  B & B & 3.49263  & 35.1382  & -31.6455  & 2.19456  & 17.2300  & -15.0354 \\[0.03in]
  B & C & 3.42122  & 34.1546  & -30.7334  & 2.57973  & 1.63840  & -13.8043 \\[0.03in]
  C & B & 1.85901  & 34.6593  & -32.8003  & 4.24977  & 15.4880  & -11.2382 \\[0.03in]
  C & C & 1.88374  & 34.1047  & -34.2210  & 4.29976  & 18.6530  & -14.3532 \\[0.03in]
  B & D & 3.20465  & 0 & 3.20465  & 2.13971  & 50.1943  & -48.0545 \\[0.03in]
  D & B & 0 & 37.4041  & -37.4041  & 5.55227  & 15.0221  & -9.46991 \\[0.03in]
  C & D & 1.86450  & 0 & 1.86450  & 3.82623  & 51.2675  & -47.4413 \\[0.03in]
  D & C &  0 & 34.2917  & -34.2917  & 5.60762  & 16.0293  & -10.4217 \\[0.03in]
  D & D & 0 & 0 & 0 & 5.95988 & 53.8175  & -47.8576 \\[0.03in]
  \hline
\end{tabular}
\label{table:CDS1}
\end{table}

\begin{table}[H]
\caption{Mitigation in the CVA of a CDS (in \%), $\alpha=0$, No collateralization}
\smallskip
\centering
\renewcommand{\arraystretch}{1.3}
\begin{tabular}{cccccccc}
\hline \hline
     (B,B)   & (B,C)    & (C,B)    &  (C,C)     & (B,D)  & (D,B)     & (C,D)  & (D,C) \\
    \hline
   68.58\%  & 71.15\%   &  76.52\%  &  70.01\%   &  -0.41\% &  80.21\%  &  0.87\% & 78.22\%  \\[0.03in]
  \hline
\end{tabular}
\label{table:CDS1g}
\end{table}

 \begin{figure}[H]
\centering
\begin{tabular}{c}
\epsfig{file=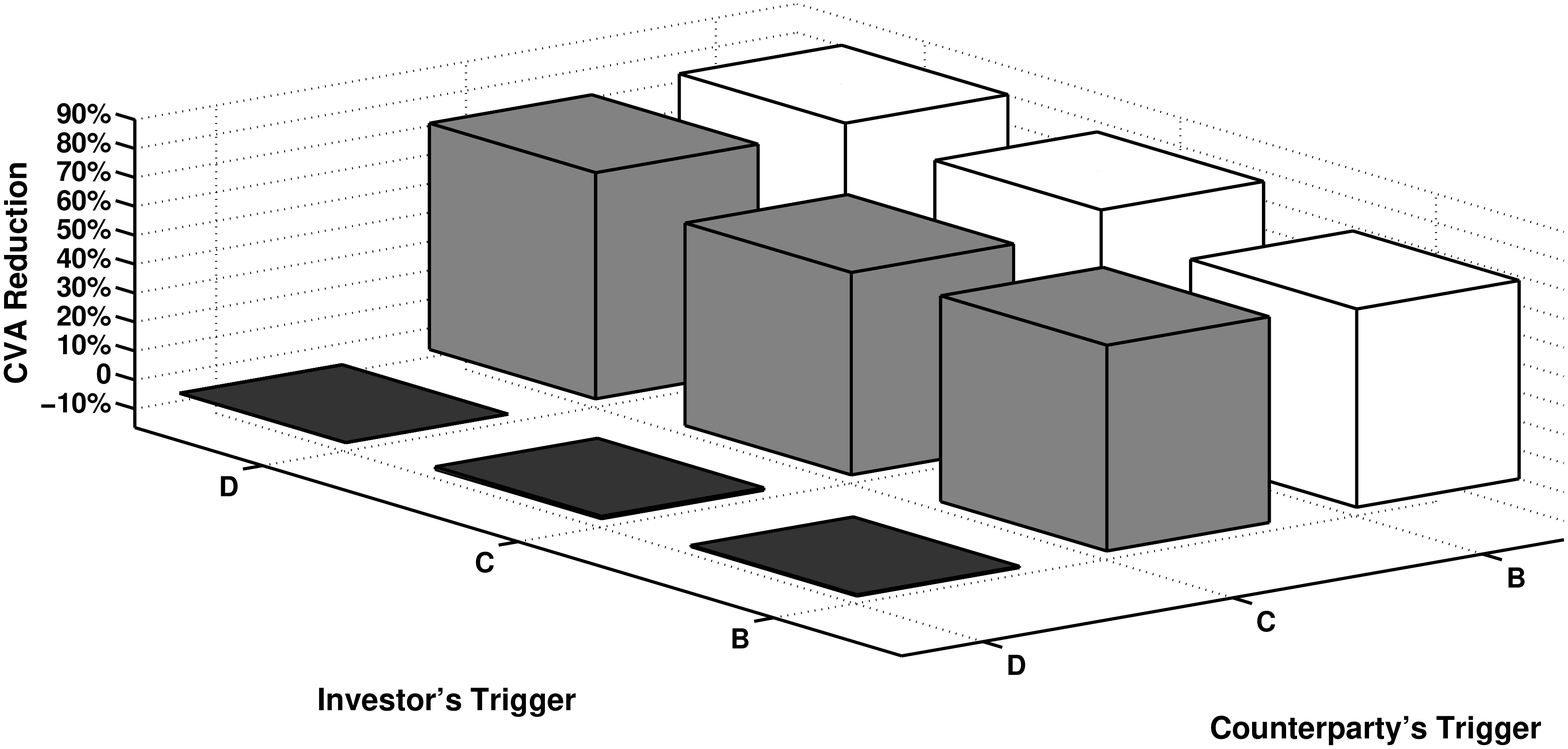,width=0.90\linewidth,clip=}
\end{tabular}
\caption{Mitigation in the CVA of a CDS (in \%), $\alpha=0$, No collateralization}
\label{CVArating4}
\end{figure}

The DRVA values are equal to zero whenever $K_2 = D$, and the URVA values are equal to zero where $K_1 = D$, since the rating triggers are set to the default levels.

\begin{table}[tH]
\caption{CVA and RVA (1\$ $\times 10^{-3}$) components of a CDS, $\alpha=1$, No collateralization}
\smallskip
\centering
\renewcommand{\arraystretch}{1.3}
\begin{tabular}[c]{cccccccc}
\hline \hline
$K_1$ & $K_2$ & URVA  & DRVA & RVA & UCVA$^R$ & DVA$^R$ & CVA$^R$ \\
    \hline
  B & B & 5.50082  & 37.8154  & -32.3146  & 2.55096  & 17.8533  & -15.3024 \\[0.03in]
  B & C & 5.10706  & 36.0968  & -30.9897  & 2.08154  & 15.6063  & -13.5248 \\[0.03in]
  C & B & 3.00767  & 37.8859  & -34.8782  & 1.94981  & 16.9493  & -14.9994 \\[0.03in]
  C & C & 3.08081  & 37.4707  & -34.3899  & 4.65988  & 17.8180  & -13.1581 \\[0.03in]
  B & D & 5.54867  & 0 & 5.54867  & 3.18332  & 47.0143  & -43.8309 \\[0.03in]
  D & B & 0 & 38.7167  & -38.7167  & 32.1677  & 17.7398  & -14.5230 \\[0.03in]
  C & D & 2.98910  & 0 & 2.98910  & 5.39572  & 50.2480  & -44.8523 \\[0.03in]
  D & C & 0 & 37.6789  & -37.6789  & 7.41472  & 18.1496  & -10.7349 \\[0.03in]
  D & D & 0 & 0 & 0 & 8.96186  & 56.8067  & -47.8448 \\[0.03in]
  \hline
\end{tabular}
\label{table:CDS2}
\end{table}

\begin{table}[H]
\caption{Mitigation in the CVA of a CDS (in \%), $\alpha=1$, No collateralization}
\smallskip
\centering
\renewcommand{\arraystretch}{1.3}
\begin{tabular}{cccccccc}
\hline \hline
   (B,B)  & (B,C)   & (C,B)  &  (C,C)  & (B,D)   & (D,B)   & (C,D)   & (D,C)  \\
    \hline
   68.02\% & 71.73\%  & 68.65\% & 72.50\%  & 8.39\% &  69.65\% & 6.25\% & 77.56\% \\[0.03in]
  \hline
\end{tabular}
\label{table:CDS2g}
\end{table}

 \begin{figure}[H]
\centering
\begin{tabular}{c}
\epsfig{file=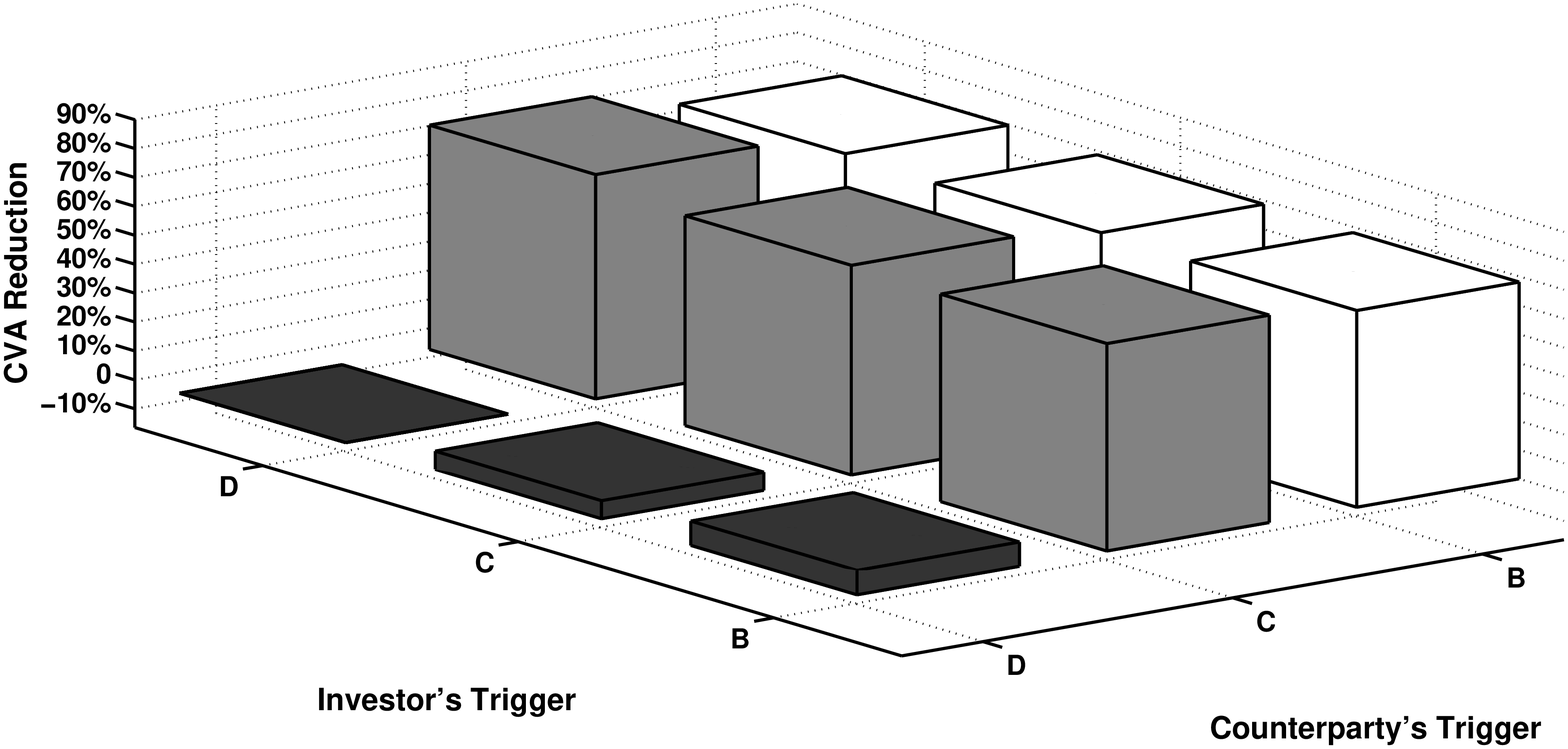,width=0.90\linewidth,clip=}
\end{tabular}
\caption{Mitigation in the CVA of a CDS (in \%), $\alpha=1$, No collateralization}
\label{CVArating4a}
\end{figure}

\begin{table}[tH]
\caption{CVA and RVA (1\$ $\times 10^{-3}$) components of a CDS, $\alpha=0$, Linear collateral rate: $\rho^i_l$}
\smallskip
\centering
\renewcommand{\arraystretch}{1.3}
\begin{tabular}[c]{cccccccc}
\hline \hline
$K_1$ & $K_2$ & URVA  & DRVA & RVA & UCVA$^R$ & DVA$^R$ & CVA$^R$ \\
    \hline
  B & B & 1.82768  & 19.3189  & -17.4912  & 1.85172  & 8.54391  & -6.69219 \\[0.03in]
  B & C & 2.20520  & 15.4545  & -13.2493  & 2.43304  & 12.2874  & -9.85440 \\[0.03in]
  C & B & 0.70377 & 19.7287  & -19.0249  & 3.23993  & 8.37492  & -5.13498 \\[0.03in]
  C & C & 0.75460  & 14.9528  & -14.1982  & 3.8625  & 13.734  & -9.8717 \\[0.03in]
  B & D & 1.75961  & 0 & 1.75961  & 2.01114  & 28.6036  & -26.5925 \\[0.03in]
  D & B & 0 & 21.3386  & -21.3386  & 3.94483  & 8.31585  & -4.37102 \\[0.03in]
  C & D & 0.83638  & 0 & 0.83638  & 3.20253  & 28.1376  & -24.9351 \\[0.03in]
  D & C & 0 & 14.3316  & -14.3316  & 3.96853  & 12.6725  & -8.70399 \\[0.03in]
  D & D & 0 & 0 & 0 & 4.08131  & 30.1705  & -26.0892 \\[0.03in]
  \hline
\end{tabular}
\label{table:CDS3}
\end{table}

\begin{table}[H]
\caption{Mitigation in the CVA of a CDS (in \%), $\alpha=0$, Linear collateral rate: $\rho^i_l$}
\smallskip
\centering
\renewcommand{\arraystretch}{1.3}
\begin{tabular}{cccccccc}
\hline \hline
     (B,B)   & (B,C)    & (C,B)    &  (C,C)     & (B,D)  & (D,B)     & (C,D)  & (D,C)  \\
    \hline
   74.35\%  & 62.23\%  &  80.32\%  &  62.16\%   &  -1.93\% &  83.25\%  &  4.42\% & 66.64\%   \\[0.03in]
  \hline
\end{tabular}
\label{table:CDS3g}
\end{table}

 \begin{figure}[H]
\centering
\begin{tabular}{c}
\epsfig{file=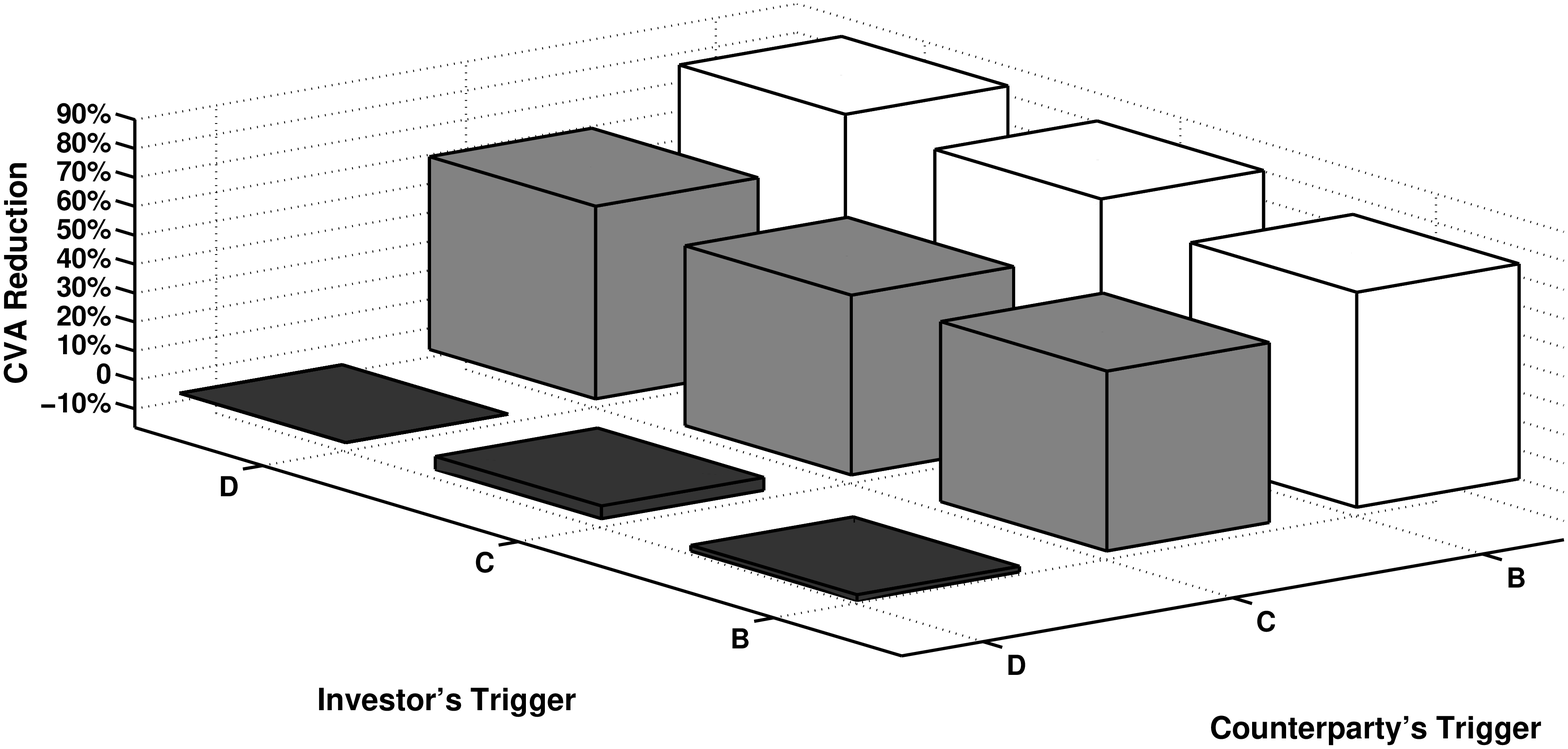,width=0.90\linewidth,clip=}
\end{tabular}
\caption{Mitigation in the CVA of a CDS (in \%), $\alpha=0$, Linear collateral rate: $\rho^i_l$}
\label{CVArating5}
\end{figure}

\begin{table}[tH]
\caption{CVA and RVA (1\$ $\times 10^{-3}$) components of a CDS, $\alpha=1$, Linear collateral rate: $\rho^i_l$}
\smallskip
\centering
\renewcommand{\arraystretch}{1.3}
\begin{tabular}[c]{cccccccc}
\hline \hline
$K_1$ & $K_2$ & URVA  & DRVA & RVA & UCVA$^R$ & DVA$^R$ & CVA$^R$ \\
    \hline
  B & B & 3.20625  & 20.8636  & -17.6574  & 2.04873  & 9.17156  & -7.12283 \\[0.03in]
  B & C & 3.20643  & 15.2936  & -12.0872  & 2.09685  & 12.1584  & -10.0616 \\[0.03in]
  C & B & 1.27267  & 22.9444  & -21.6717  & 2.03668  & 8.75374  & -6.71706 \\[0.03in]
  C & C & 1.10322  & 15.9222  & -14.8190  & 3.53085  & 12.8222  & -9.29137 \\[0.03in]
  B & D & 3.20643  & 0 & 3.20643  & 2.46990  & 25.4939  & -23.0240 \\[0.03in]
  D & B & 0 & 22.9444  & -22.9444  & 2.18442 & 8.81189  & -6.62742 \\[0.03in]
  C & D & 1.38153  & 0 & 1.38153  & 4.55091  & 27.8220  & -23.2711 \\[0.03in]
  D & C & 0 & 16.3974  & -16.3974  & 5.33103  & 14.4275  & -9.09652 \\[0.03in]
  D & D & 0 & 0 & 0 & 5.86474  & 29.8551  & -23.9904 \\[0.03in]
  \hline
\end{tabular}
\label{table:CDS4}
\end{table}

\begin{table}[H]
\caption{Mitigation in the CVA of a CDS (in \%), $\alpha=1$, Linear collateral rate: $\rho^i_l$}
\smallskip
\centering
\renewcommand{\arraystretch}{1.3}
\begin{tabular}{cccccccc}
\hline \hline
   (B,B)  & (B,C)   & (C,B)  &  (C,C)  & (B,D)   & (D,B)   & (C,D)   & (D,C)  \\
    \hline
   70.31\% & 58.06\%  & 72\% & 61.27\%  & 4.03\% &  72.37\% & 3\% & 62.08\%  \\[0.03in]
  \hline
\end{tabular}
\label{table:CDS4g}
\end{table}

 \begin{figure}[H]
\centering
\begin{tabular}{c}
\epsfig{file=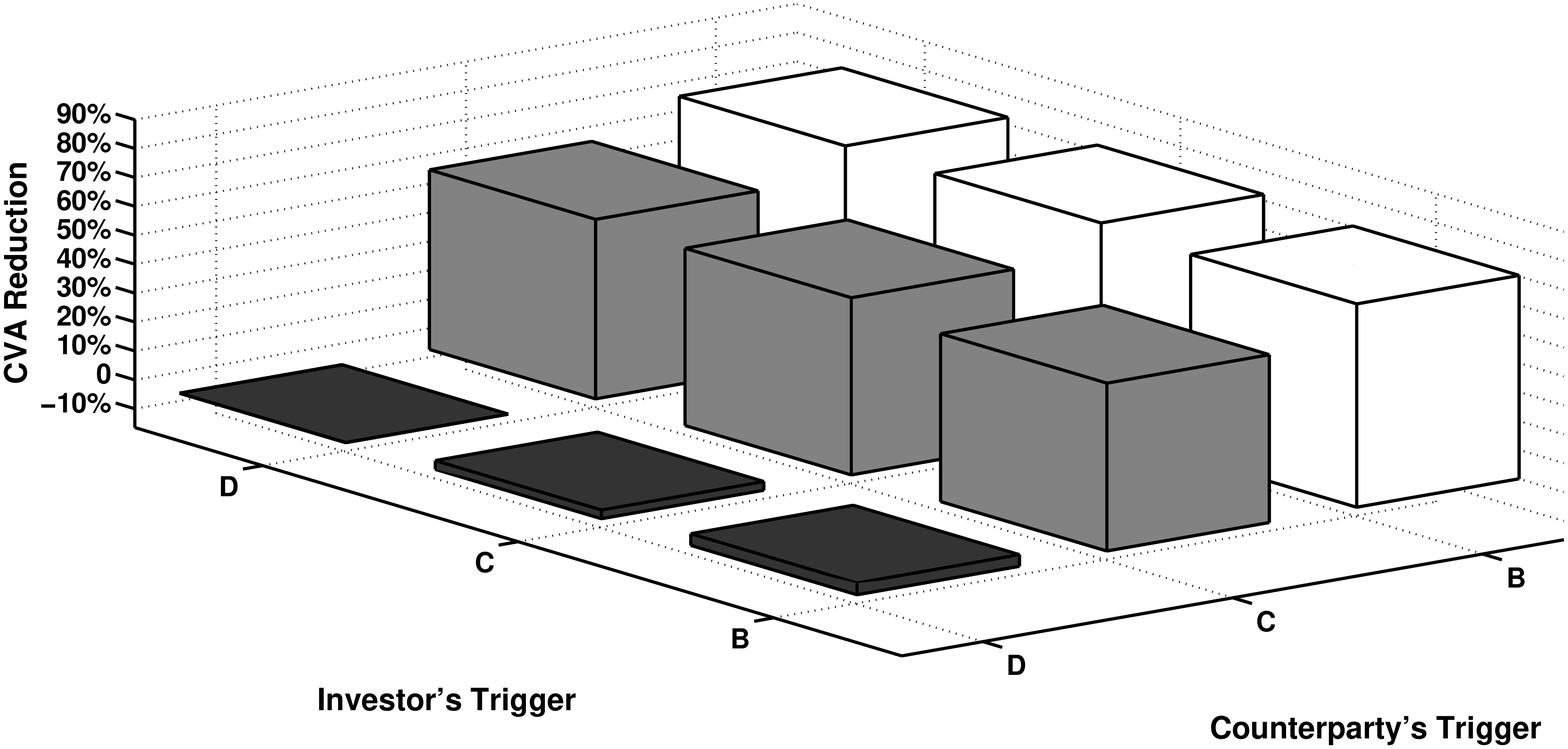,width=0.90\linewidth,clip=}
\end{tabular}
\caption{Mitigation in the CVA of a CDS (in \%), $\alpha=1$, Linear collateral rate: $\rho^i_l$}
\label{CVArating5a}
\end{figure}

\begin{table}[tH]
\caption{CVA and RVA (1\$ $\times 10^{-3}$) components of a CDS, $\alpha=0$, Exponential collateral rate: $\rho^i_e$}
\smallskip
\centering
\renewcommand{\arraystretch}{1.3}
\begin{tabular}[c]{cccccccc}
\hline \hline
$K_1$ & $K_2$ & URVA  & DRVA & RVA & UCVA$^R$ & DVA$^R$ & CVA$^R$ \\
    \hline
  B & B & 0.98820  & 11.0877  & -10.0994  & 1.79938  & 6.15293  & -4.35355 \\[0.03in]
  B & C & 1.26294  & 7.94327  & -6.68033  & 2.33252  & 9.20794  & -6.87542 \\[0.03in]
  C & B & 0.33328  & 11.3906  & -11.0573  & 2.71652  & 6.23845  & -3.52192 \\[0.03in]
  C & C & 0.34274  & 7.53087  & -7.18812  & 3.14017  & 10.4728  & -7.33268 \\[0.03in]
  B & D & 0.98106  & 0 & 0.98106  & 1.93606  & 18.0396  & -16.1036 \\[0.03in]
  D & B & 0 & 12.6680  & -12.6680  & 2.98248  & 6.16772  & -3.18523 \\[0.03in]
  C & D & 0.39375  & 0 & 0.39375  & 2.75329  & 17.6440  & -14.8907 \\[0.03in]
  D & C & 0 & 7.15044  & -7.15044  & 3.01623  & 9.48455  & -6.46832 \\[0.03in]
  D & D & 0 & 0 & 0 & 3.19195  & 18.7888  & -15.5969 \\[0.03in]
  \hline
\end{tabular}
\label{table:CDS5}
\end{table}

\begin{table}[H]
\caption{Mitigation in the CVA of a CDS (in \%), $\alpha=0$, Exponential collateral rate: $\rho^i_e$}
\smallskip
\centering
\renewcommand{\arraystretch}{1.3}
\begin{tabular}{cccccccc}
\hline \hline
     (B,B)   & (B,C)    & (C,B)    &  (C,C)    & (B,D) & (D,B)     & (C,D)  & (D,C)  \\
    \hline
   72.09\%  & 55.92\%  &  77.42\%  & 52.98\%  & -3.25\% & 79.58\%   & 4.53\% & 58.53\%  \\[0.03in]
  \hline
\end{tabular}
\label{table:CDS5g}
\end{table}

 \begin{figure}[H]
\centering
\begin{tabular}{c}
\epsfig{file=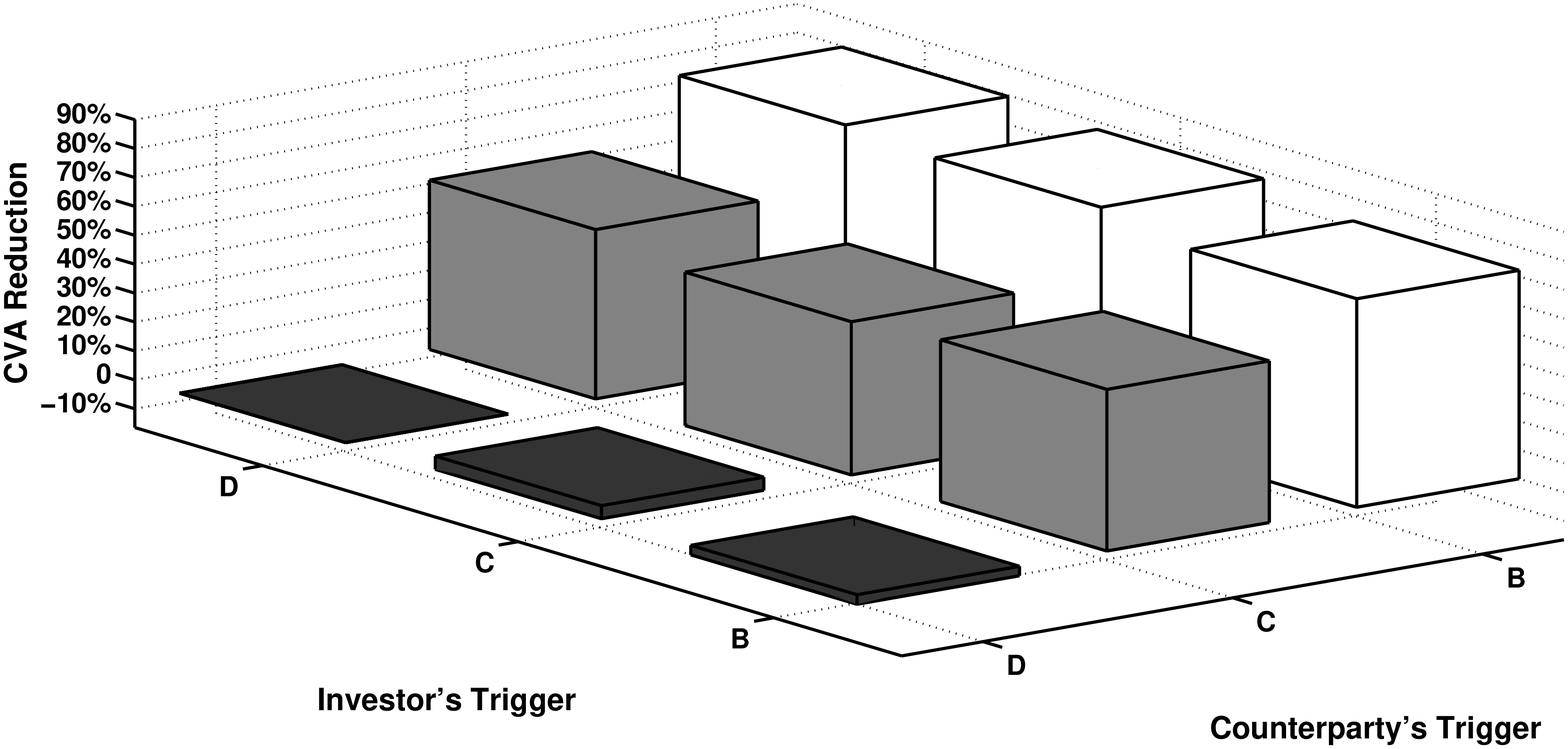,width=0.90\linewidth,clip=}
\end{tabular}
\caption{Mitigation in the CVA of a CDS (in \%), $\alpha=0$, Exponential collateral rate: $\rho^i_e$}
\label{CVArating6}
\end{figure}

\begin{table}[tH]
\caption{CVA and RVA (1\$ $\times 10^{-3}$) components of a CDS, $\alpha=1$, Exponential collateral rate: $\rho^i_e$}
\smallskip
\centering
\renewcommand{\arraystretch}{1.3}
\begin{tabular}[c]{cccccccc}
\hline \hline
$K_1$ & $K_2$ & URVA  & DRVA & RVA & UCVA$^R$ & DVA$^R$ & CVA$^R$ \\
    \hline
  B & B & 1.80010  & 11.7752  & -9.97518  & 2.84797  & 6.92726  & -4.07929 \\[0.03in]
  B & C & 1.88120  & 8.29421  & -6.41300  & 2.17836  & 8.83311  & -6.65475 \\[0.03in]
  C & B & 0.64039  & 13.1186  & -12.4782  & 2.07924  & 6.89925  & -4.82001 \\[0.03in]
  C & C & 0.62521  & 8.91043  & -8.28522  & 3.26987  & 10.3745  & -7.10464 \\[0.03in]
  B & D & 1.69567  & 0 & 1.69567  & 2.62547  & 15.2304  & -12.6049 \\[0.03in]
  D & B & 0 & 11.7182  & -11.7182  & 2.76178  & 7.08910  & -4.32731 \\[0.03in]
  C & D & 0.747189  & 0 & 0.747189  & 4.13449  & 16.4391  & -12.3047 \\[0.03in]
  D & C & 0 & 8.78431  & -8.78431  & 3.61762  & 10.3593  & -6.74173 \\[0.03in]
  D & D & 0 & 0 & 0 & 4.08585  & 19.2965  & -15.2107 \\[0.03in]
  \hline
\end{tabular}
\label{table:CDS6}
\end{table}

\begin{table}[H]
\caption{Mitigation in the CVA of a CDS (in \%), $\alpha=1$, Exponential collateral rate: $\rho^i_e$}
\smallskip
\centering
\renewcommand{\arraystretch}{1.3}
\begin{tabular}{cccccccc}
\hline \hline
   (B,B)  & (B,C)   & (C,B)  &  (C,C)  & (B,D)   & (D,B)   & (C,D)   & (D,C)  \\
    \hline
   73.18\% & 56.25\%  & 68.31\% & 53.29\%  & 17.13\% &  71.55\% & 19.11\% & 55.68\% \\[0.03in]
  \hline
\end{tabular}
\label{table:CDS6g}
\end{table}

\begin{figure}[H]
\centering
\begin{tabular}{c}
\epsfig{file=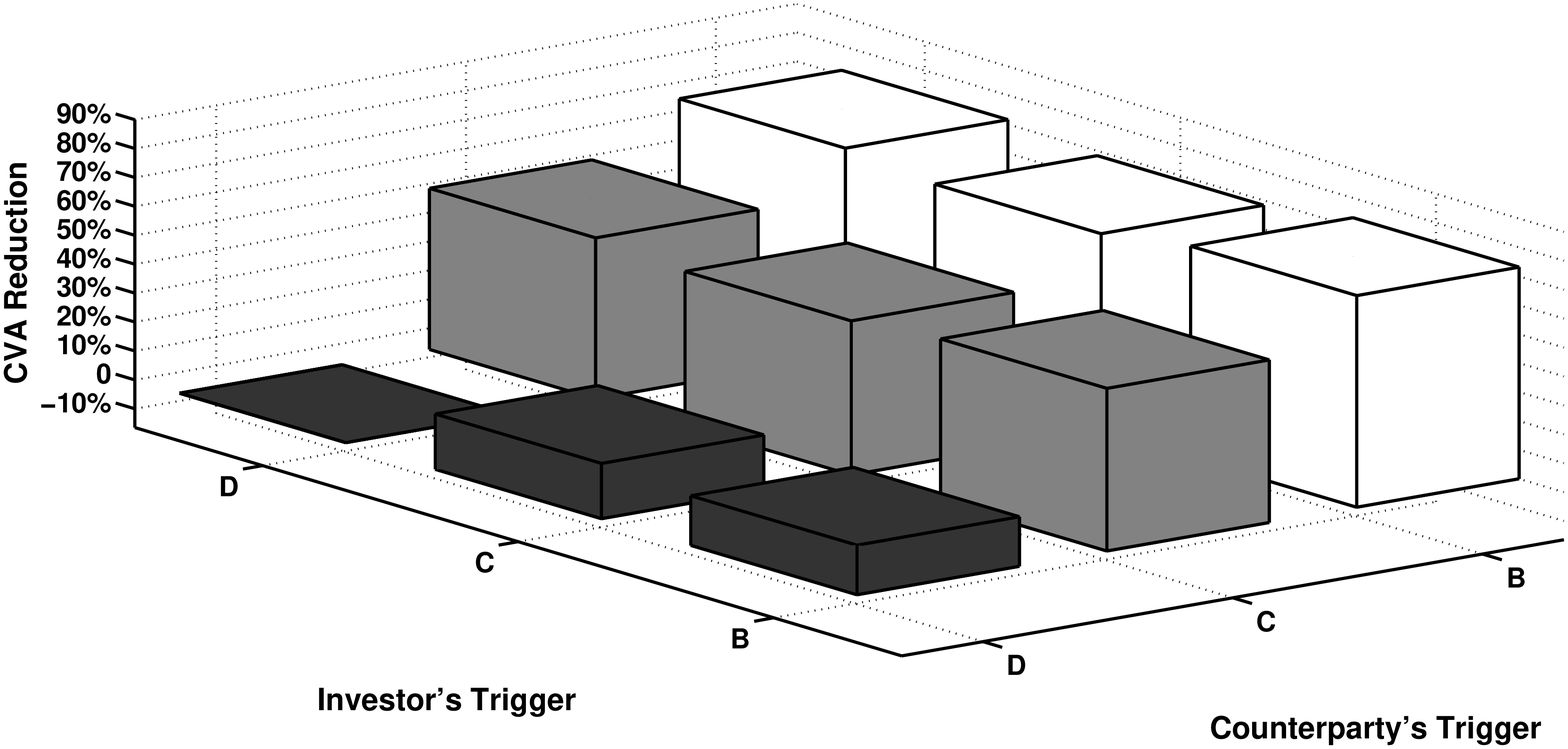,width=0.90\linewidth,clip=}
\end{tabular}
\caption{Mitigation in the CVA of a CDS (in \%) (in \%), $\alpha=1$, Exponential collateral rate: $\rho^i_e$}
\label{CVArating6a}
\end{figure}

\bibliographystyle{alpha}

\end{document}